\documentclass[times,twocolumn,final]{elsarticle}
\usepackage{cag}
\usepackage{framed,multirow}

\usepackage{amssymb}
\usepackage{latexsym}

\usepackage{url}
\usepackage{verbatim}
\usepackage{times}
\usepackage{epsfig}
\usepackage{graphicx}
\usepackage{amsmath,amsfonts,amssymb,amsthm,version}
\usepackage{algorithmic}
\usepackage{algorithm}
\usepackage{appendix}
\usepackage{array}
\usepackage{tikz-cd}
\usepackage[caption=false,font=normalsize,labelfont=sf,textfont=sf]{subfig}

\newtheorem{theorem}{Theorem}
\newtheorem{lemma}{Lemma}

\newtheorem{property}{Property}
\newtheorem{definition}{Definition}
\newtheorem{example}{Example}

\theoremstyle{remark}
\newtheorem{remark}{Remark}

\usepackage{xcolor}
\definecolor{newcolor}{rgb}{.8,.349,.1}

\usepackage{hyperref}

\usepackage[switch,pagewise]{lineno} %Required by command \linenumbers below

\journal{Computers \& Graphics}

\begin{document}

\verso{Preprint Submitted for review}

\begin{frontmatter}

\title{Analyzing Singular Patterns in Discrete Planar Vector Fields via Persistent Path Homology}%

\author[1]{Yu \snm{Chen}}
\ead{chenyu.math@zju.edu.cn}
    
\author[1]{Hongwei \snm{Lin} \corref{cor1}}
\cortext[cor1]{Corresponding author}
\emailauthor{hwlin@zju.edu.cn}{Hongwei Lin}

\address[1]{School of Mathematics Science, Zhejiang University, Hangzhou, 310058, China}

%\author[1]{First Author Given Name \snm{Surname}\corref{cor1}}
%\cortext[cor1]{Corresponding author: 
%  Tel.: +0-000-000-0000;  
%  fax: +0-000-000-0000;}
%\emailauthor{example@email.com}{Corresponding Author Name}
%\ead{example@email.com}
    
%\author[2]{Second Author Given Name \snm{Surname}\fnref{fn1}}
%\fntext[fn1]{Footnote 1.}  

%\address[1]{Address, City, Postcode, Country}
%\address[2]{Address, City, Postcode, Country}

%\received{1 February 2017}
%\received{\today}
%%%% Do not use the below for submitted manuscripts
%\finalform{28 March 2017}
%\accepted{2 April 2017}
%\availableonline{15 May 2017}
%\communicated{S. Sarkar}

\begin{abstract}
Analyzing singular patterns in vector fields is a fundamental problem in theoretical and practical domains due to the ability of such patterns to detect the intrinsic characteristics of vector fields. In this study, we propose an approach for analyzing singular patterns from discrete planar vector fields. Our method involves converting the planar discrete vector field into a specialized digraph and computing its one-dimensional persistent path homology. By analyzing the persistence diagram, we can determine the location of singularities, and the variations of singular patterns can also be analyzed. The experimental results demonstrate the effectiveness of our method in analyzing the singular patterns of noisy real-world vector fields and measuring the variations between different vector fields.
\end{abstract}


\end{frontmatter}

%\linenumbers

\section{Introduction}\label{sec:1}
Analyzing singular patterns in vector fields is an important problem in theoretical and practical domains. Singular patterns can detect the intrinsic characteristics of vector fields, such as divergence (convergence) patterns corresponding to sources (sinks) and rotational patterns corresponding to center points or spiral sources (sinks). They are commonly found in various vector field analysis problems, including vortex features of electromagnetic fields, vortex features generated by water currents, and wind field features of tropical cyclones. Identifying and extracting these singular patterns from vector fields is crucial to theoretical understanding and practical applications \cite{wong2009identifying,li2006segmentation}.
In practical scenarios, the explicit expression of the vector field is frequently unknown, and only the field vectors at certain uniformly sampled grid points are accessible. In such instances, these vector fields are considered discrete. 

While there exists a long history of research on the analysis of singular patterns, classical numerical methods may encounter accuracy issues and lack the capability to compare differences between singular patterns. In this study, we propose a new method to determine the position of singularities. Moreover, we can analyze the variations of singular patterns of a time-varying vector field. Given that discrete vector fields are similar to grid-shaped digraphs, we convert the grid-sampled vector field $X$ into a specialized digraph and analyze its features using persistent path homology. Path homology \cite{grigor2012homologies}, also known as GLMY homology, is a powerful tool for analyzing digraphs because it is sensitive to the directions in digraphs. Moreover, persistent path homology is a new perspective within the field of computational topology; it is specifically designed for analyzing digraph models. If a digraph contains intrinsic features conveyed by the direction and weight of its edges, these features can be discovered using (persistent) path homology.

Specifically, our method involves the following steps: First, we create a special digraph called an angle-based grid digraph from the vector field $X$ and construct the corresponding digraph filtration. Then, the one-dimensional persistent path homology is computed to obtain the persistence diagram. Finally, the position of singularities can be determined based on the information provided by the persistence diagram. Additionally, suppose a time-varying vector field is given. In that case, we can also measure the changes in the topological features of the vector field by computing the distance between persistence diagrams. 
In our experiments, the proposed method is employed to analyze geomagnetic fields and tropical cyclones by converting their real-world planar vector field data into a digraph and subsequently extracting the corresponding singular patterns using persistent path homology.

In summary, the main contributions of our study are as follows:
\begin{itemize}
    \item We develop a method to convert discrete vector fields into digraphs while preserving the information of the singular patterns.
    \item The positions of the singularities are determined by computing the one-dimensional persistent path homology of the digraph filtration.
    \item For a time-varying vector field, the variations can be measured by computing the distance between the corresponding persistence diagrams.
\end{itemize}

The remainder of this paper is organized as follows: We provide an overview of persistent path homology and discrete vector field analysis in Section \ref{sec:2}. Then, we present a brief preliminary of vector field, digraph, and path homology in Section \ref{sec:3}. The proposed method is introduced in Section \ref{sec:4} and Section \ref{sec:5}. In Section \ref{sec:6} we show our experiment results and provide a discussion. Section \ref{sec:7} presents the conclusions and future works.

\section{Related work}\label{sec:2}
In this section, we provide an overview of persistent path homology and discrete vector field analysis. We first introduce persistent path homology and then discuss various methods for analyzing discrete vector fields, including the detection of singularities and the segmentation of singular patterns. Additionally, a brief introduction to comparing vector fields is provided.

\subsection{Persistent path homology with applications}
Persistent homology is a powerful tool in topological data analysis for finding topological structures of data during filtration \cite{edelsbrunner2008persistent,zomorodian2004computing,edelsbrunner2002topological}. Path homology, proposed by Grigor’yan et al. \cite{grigor2012homologies}, further extends the concept of persistent homology. It has been developed into various theories, such as homotopy theory for digraphs \cite{grigor2014homotopy}, discrete Morse theory on digraphs \cite{lin2021discrete}, and path homology theory of multigraphs and quivers \cite{grigor2018path}. Chowdhury and Mémoli \cite{chowdhury2018persistent} introduced the concept of persistent path homology of networks, which provides a new method for handling digraph models in the category of topological data analysis. Dey et al. \cite{dey2022efficient} proposed an efficient algorithm for computing one-dimensional persistent path homology of digraph filtrations. 

Path homology has found successful application in various practical problems. Chowdhury et al. \cite{chowdhury2019path} utilized path homology to investigate the structure of deep feedforward networks. Chen et al. \cite{chen2023path} revealed intrinsic mathematical characteristics of molecules and materials by constructing directed networks and elucidated corresponding functional structures using persistent path homology. Liu et al. \cite{liu2023neighborhood} proposed a structural characterization method based on path homology theory to extract structural information from materials and predict their properties. Wu et al. \cite{wu2023metabolomic} employed path homology theory to analyze and interpret the topological changes in health states from symbiosis to dysbiosis and vice versa, identifying several hub metabolites and their interaction webs. In our study, we utilize one-dimensional path homology as the main tool to detect the singular patterns in the vector fields.

\subsection{Discrete vector field analysis}
The study of discrete vector field topology has a long history, with numerous techniques available for analyzing vector field topology. Singularities and the flow curves or surfaces that connect them are classical tools for visualizing the topology of a vector field \cite{helman1991visualizing}. Additionally, various methods exist for addressing different types of singularities in vector fields and analyzing their impact \cite{post2003state,laramee2007topology,salzbrunn2008state,pobitzer2011state}, providing different ways for visualizing vector fields. Moreover, vortices, which represent more general and physics-related features in vector fields and commonly appear in practical scenarios, can also be identified using a wide range of techniques designed for vector fields \cite{gunther2018state}.

We now discuss methods for detecting singularities in discrete vector fields. If the expression of the vector field is known, the position of the singularity can be found by solving the system of equations \cite{andronov1974qualitative}. However, this method is not applicable to discrete vector fields. Therefore, methods specifically designed for the discrete case have been proposed. First, there are some numerical methods such as triangular linear interpolation of vectors \cite{tricoche2002topology} and the computation of the Jacobian matrix of the discrete vector field \cite{helman1989representation}. Another useful technique is Hodge decomposition \cite{polthier2003identifying,tong2003discrete}, which decomposes vector fields into three intuitive components: a divergence-free part, a curl-free part, and a harmonic part. Singularities can be detected using the decomposed vectors. Moreover, Wong and Yip \cite{wong2009identifying} introduced a method for finding the centers of circulating and spiraling vector field patterns by fitting a logarithmic spiral equation and modeling the field vectors for singular patterns corresponding to center points or spiral sources (sinks). 
Crane et al. \cite{crane2010trivial} proposed a connection defined on cells that can be used to design vector fields, and its underlying ideas can also be applied to find the locations of singularities. This method assigns vectors to each cell, and the singularities are considered to be located at the vertices of some cells. It is also possible to track singularities of time-varying vector fields.
But our study focuses on detecting the accurate position of singularities within the cells. Therefore, the main purpose in \cite{crane2010trivial} is different from our target.
In contrast, unlike certain numerical methods that may fail to find singularities or encounter spurious singularities caused by precision issues, thus demonstrating their limitations, this study proposes a method based on persistent path homology and an angle-based grid digraph for detecting singular patterns, relying on the singularity index.

Another issue arises in comparing vector fields from a topological perspective and measuring their differences, particularly investigating the variations of singular patterns in time-varying vector fields. Theisel et al. \cite{Theisel2003UsingFF} introduced a topology-based comparison method that identifies changes in the singularities present in vector fields, serving as the primary method currently employed for comparing vector fields. Moreover, other topological features, such as closed orbits and separatrices, are also utilized in comparing vector fields at different moments in time-varying planar vector fields \cite{tricoche2002topology}. But these methods are hard to measure global variations of singular patterns between vector fields. In our study, we utilize the distance between persistence diagrams to measure the topological variations of vector fields, particularly time-varying vector fields.

\section{Preliminaries}\label{sec:3}
\subsection{Concepts of vector fields}
In the remainder of this paper, we focus on planar vector fields. We introduce some basic concepts about planar vector fields, which exhibit important topological properties. References on vector field properties can be found in \cite{andronov1974qualitative}.

\begin{definition}(Singularity)
    Given a planar vector field $ X=(P(x,y),Q(x,y)) $ defined on $D \subseteq \mathbb{R}^2$, then a \textit{singularity} $(x_0,y_0)\in D$ of $ X $ is a point such that $P(x_0,y_0)=Q(x_0,y_0)=0$.
\end{definition}
Singularity is an important concept of vector fields because it can determine the topologies of vector fields. From the perspective of the flow field or dynamical system, we can classify singularities into different types, such as \textit{saddle}, \textit{source}, \textit{sink}, \textit{center point}, \textit{spiral source}, and \textit{spiral sink} \cite{andronov1974qualitative}.

Consider a simple closed curve $L$ that does not contain any singularities of the vector field $X=(P,Q)$. The \textit{winding number} of $X$ with respect to $L$ is defined as
$$
j = \frac{1}{2\pi} \oint_L d\arctan\left(\frac{Q}{P}\right) = \frac{1}{2\pi} \oint_L \frac{P\,dQ - Q\,dP}{P^2+Q^2},
$$
where the integration along $L$ is performed counterclockwise. 
When the region bounded by $L$ contains exactly one singularity $S(x_0,y_0)$ of $X$, this winding number is called the \textit{index} of $S$. 
For a continuous planar vector field, the winding number (or index) satisfies several fundamental properties: (1)For any simple closed curve, the index equals the sum of indices of all enclosed singularities.
(2) A simple closed curve enclosing no singularities has index $0$.

In this study, we focus on the vector fields that have no more than one singularity in each square (which will be introduced later), and the index of each singularity is either 1 or -1 only.

\subsection{Digraph and its path homology}
\begin{definition}(Digraph)
    A \textit{digraph} is an ordered pair $G =(V,E)$, where $ V $ is a set of all vertices and $ E $ is a set of ordered pairs of vertices, i.e. directed edges that satisfy $E \subseteq V\times V$. If $G =(V,E)$ does not contain any self-loop, it is called a \textit{simple digraph}.
\end{definition}
For example, the digraph
$$\begin{tikzcd}
a \arrow[r] \arrow[r] \arrow[r] & b \arrow[r] & c \arrow[loop, distance=2em, in=305, out=235]
\end{tikzcd}
$$
is not a simple digraph, since it has a self-loop at node $c$.

Now we introduce the path homology of digraphs. The theories of path homology of general path complexes can be found in previous works \cite{grigor2012homologies}. 
\begin{definition}(Elementary $ n $-path and boundary operator)
    Let $V$ be an arbitrary nonempty finite set whose elements are called vertices. For $n \ge 0$, an \textit{elementary $n$-path} $i_0 \ldots i_n$ on $V$ is a sequence of $n+1$ vertices in $V$. For $n=-1$ the set of elementary $n$-path is empty set $\emptyset$. An elementary $n$-path $i_0 \ldots i_n$ is also denoted by $e_{i_0 \ldots i_n}$. The boundary operator is defined as
    $$ \partial e_{i_0 \ldots i_n} = \sum_{k=0}^{n} (-1)^k e_{i_0 \ldots \hat{i}_k \ldots i_n}$$
    where the hat $\hat{i}_k$ means omission of the index $i_k$.
\end{definition}
For example, for the combination of two path $e_{023}-e_{013}$, we have
$$\begin{aligned}
    \partial(e_{023}-e_{013} )&=e_{23}-e_{03}+e_{02}-(e_{13}-e_{03}+e_{01})\\
        &=e_{23}+e_{02}-e_{13}-e_{01}.
\end{aligned}$$

An elementary path $e_{i_0 \ldots i_n}$ on a set $V$ is \textit{regular} if $i_{k-1} \neq i_k, \forall k=1, \ldots, n$. Otherwise it is \textit{non-regular}. Thus, using the vertex set of a simple digraph, we can obtain the set of elementary $n$-paths of $G$, where all elements are regular paths. In the following discussion, unless otherwise noted, simple digraphs are considered.

\begin{definition}(Allowed $ n $-paths)
    Let $ G=(V,E) $ be a simple digraph, for any $n \ge 0$, the $\mathbb{K}$-linear space $\mathcal{A}_n(G)=\mathcal{A}_n(V,E;\mathbb{K})$ is defined as
    $$\mathcal{A}_n=\mathcal{A}_n(G)=\operatorname{span}\left\{e_{i_0 \ldots i_n}: i_0, \ldots, i_n \in V, i_k i_{k+1}\in E \right\} $$
    where $k=0,1,\cdots,n-1$. The elements of $\mathcal{A}_n$ are called \textit{allowed $ n $-paths}.
\end{definition}

\begin{definition}($\partial$-invariant $ n $-paths)
    Let $ G=(V,E) $ be a simple digraph, for any $n\ge -1$, we define a subspace of $\mathcal{A}_n$ as:
    $$
    \Omega_n=\Omega_n(G)=\left\{v \in \mathcal{A}_n: \partial v \in \mathcal{A}_{n-1}\right\}
    $$
    Here $\partial$ is the boundary operator, and we define $\Omega_{-1}=\mathcal{A}_{-1} \cong \mathbb{K}$ and $\Omega_{-2}=\mathcal{A}_{-2}=\{0\}$. The elements of $\Omega_n$ are called \textit{$\partial$-invariant $ n $-paths}.
\end{definition}

\begin{definition}(Path homology groups of digraph)
    Let $G =(V,E)$ be a digraph, and we have the following chain complex of $V$ denoted as $\Omega_*(V)=\left\{\Omega_n\right\}$,
    $$
    \cdots \stackrel{\partial}{\longrightarrow} \Omega_3 \stackrel{\partial}{\longrightarrow} \Omega_2 \stackrel{\partial}{\longrightarrow} \Omega_1 \stackrel{\partial}{\longrightarrow} \Omega_0 \stackrel{\partial}{\longrightarrow} \mathbb{K} \stackrel{\partial}{\longrightarrow} 0,
    $$
    and the associated \textit{$ n $-dimensional path homology group} of $G=(V, E)$ is defined as:
    \begin{equation*}
        H_n(G)=H_n(V, E ; \mathbb{K}):=\operatorname{Ker}\left(\left.\partial\right|_{\Omega_n}\right) / \operatorname{Im}\left(\left.\partial\right|_{\Omega_{n+1}}\right) 
    \end{equation*}
\end{definition}
The elements of $Z_n:=\operatorname{Ker}\left(\left.\partial\right|_{\Omega_n}\right)$ are called \textit{$ n $-cycles}, and the elements of $B_n:=\operatorname{Im}\left(\left.\partial\right|_{\Omega_{n+1}}\right)$ are called \textit{$ n $-boundaries}. A \textit{representation} of a generator of $H_n(G)$ is an $n$-cycle in $G$ but not an $n$-boundary. Usually it is not unique.

\begin{example}
    Consider the following digraph: 
    $$\begin{tikzcd}
        a                       & c \arrow[l]                     & f \arrow[l] \\
        b \arrow[u, shift left] & d \arrow[l] \arrow[u] \arrow[r] & g \arrow[u]
    \end{tikzcd}$$
    It is a simple digraph, and there are five allowed 2-paths $e_{dba}$, $e_{dca}$, $e_{dgf}$, $e_{fca}$, $e_{gfc}$
    that generate $\mathcal{A}_2$. But there is only one $\partial$-invariant 2-path $e_{dba}-e_{dca}$ that generates $\Omega_2$, since 
    $$\begin{aligned}
        \partial(e_{dba}-e_{dca} )&=e_{ba}-e_{da}+e_{db}-(e_{ca}-e_{da}+e_{dc})\\
        &=e_{ba}+e_{db}-e_{ca}-e_{dc} \in \mathcal{A}_1.
    \end{aligned}
    $$
    And other 2-paths are not elements in $\Omega_2$. Now consider cycles $C_1=e_{db}+e_{ba}-e_{ca}-e_{dc}$ and $C_2=e_{dg}+e_{gf}+e_{fc}-e_{dc}$, then $C_1, C_2$ are cycles that generate $Z_1$. But only $C_1$ is a 1-boundary which generates $B_1$, since $C_1=\partial(e_{dba}-e_{dca})$. Hence, $H_1$ is generated by the class homologous to $C_2$ and $\dim H_1=1$.
\end{example}

Then we introduce the persistent path homology, which can be derived from classical persistent homology theory.
\begin{definition} \label{def:PPH}(One-dimensional persistent path homology \cite{dey2022efficient})
    Let $G=(V,E,w)$ be a weighted digraph where $ V $ is the vertex set, $ E $ is the edge set, and $ w $ is the weight function $w: E \to \mathbb{R}_+$. Denote $G^{\delta}=(V^{\delta}=V,E^{\delta}=\left\lbrace e\in E: w(e) \le \delta \right\rbrace )$, then we derive a \textit{digraph filtration} $ \left\lbrace G^{\delta} \hookrightarrow  G^{\delta^{\prime}} \right\rbrace_ {\delta \le \delta^{\prime} \in \mathbb{R}_+} $ created by $G$. The \textit{one-dimensional persistent path homology} of a weighted digraph $G=(V,E,w)$ is defined as the persistent vector space 
    $$\mathbb{H}_1:=\left\lbrace H_1(G^{\delta}) \xrightarrow[]{i_{\delta,\delta^{\prime}}}  H_1(G^{\delta^{\prime}}) \right\rbrace_ {\delta \le \delta^{\prime} \in \mathbb{R}_+}$$
    where $i_{\delta,\delta^{\prime}}:G^{\delta} \hookrightarrow  G^{\delta^{\prime}}$ is the natural inclusion map. And $G^{\delta}$ is also called the \textit{moment} corresponding to $\delta$ in the filtration.
\end{definition}
In other words, a digraph filtration created by $G$ can be defined as a subgraph sequence that retains all vertices of $G$ and gradually adds edges of $G$ in a non-decreasing order of weights.

A valuable tool for visualizing persistent homology is the \textit{persistence diagram} (PD). The points $(b_i, d_i)$ in the persistence diagram are referred to as \textit{persistence pairs}, where $b_i$ represents the birth time of a homology generator and $d_i$ represents the death time of that homology generator. Fig. \ref{fig:PD} illustrates an example of a persistence diagram. Moreover, various distances are used to quantify the dissimilarities of two persistence diagrams, such as bottleneck distance and Wasserstein distance \cite{edelsbrunner2008persistent,cohen2005stability,chazal2008stability,cohen2010lipschitz}. These distances can measure the topological differences between the original datasets. The persistence diagram can also be utilized to describe persistent path homology. We denote the one-dimensional path persistence diagram of $\mathbb{H}_1$ as $Dgm_1(G)$.
\begin{figure}[!t]
    \centering
    \includegraphics[width=0.95\linewidth]{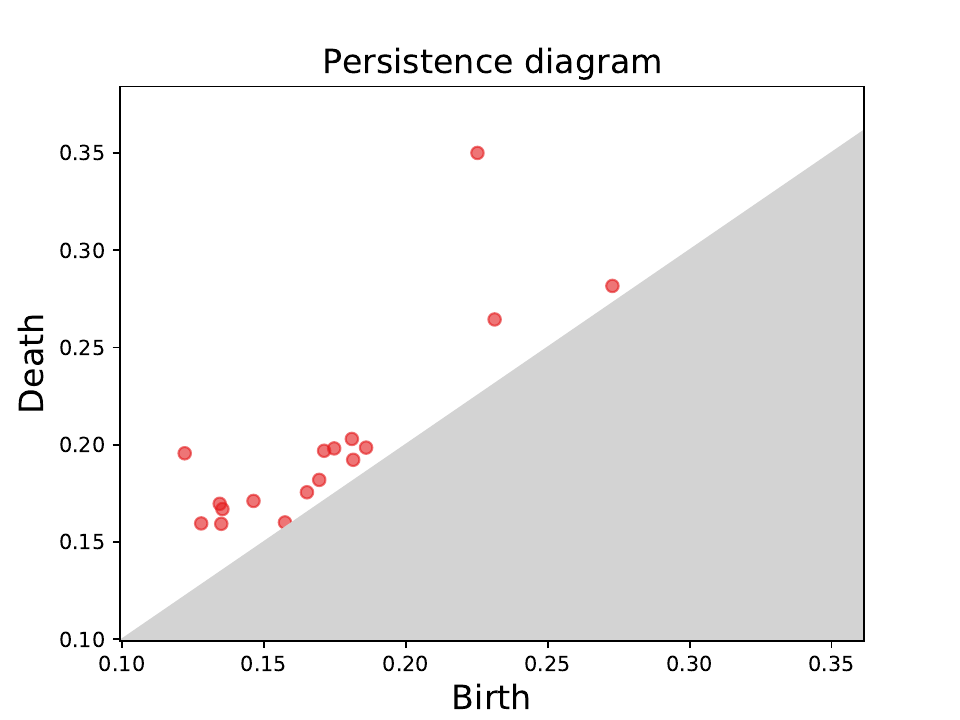}
    \caption{ Example of a persistence diagram.
    }
    \label{fig:PD}
\end{figure}

We now introduce a special kind of digraph, called the \textit{cycle digraph}, which is useful when constructing a digraph from a vector field in the subsequent sections.
\begin{definition}(Cycle digraph \cite{grigor2012homologies})
    A digraph $G = (V,E)$ is a \textit{cycle digraph} if it is connected (as an undigraph), and every vertex has the degree 2.
\end{definition}
A \textit{bigon} is a sequence of two distinct vertices $a, b \in V$ such that $a\to b, b\to a$.

A \textit{(boundary) triangle} is a sequence of three distinct vertices $a, b, c \in V$ such that $a\to b, b\to c, a\to c$:
$$\begin{tikzcd}
    & b \arrow[rd] &   \\
    a \arrow[ru] \arrow[rr] &              & c
\end{tikzcd}$$

A \textit{(boundary) square} is a sequence of four distinct vertices $a, b, c, d \in V$ such that $a\to b, b\to c, a\to d, d\to c$:
$$\begin{tikzcd}
    b \arrow[r]           & c           \\
    a \arrow[u] \arrow[r] & d \arrow[u]
\end{tikzcd}$$
A bigon, a triangle, and a square are all cycle digraphs. 

\begin{theorem}(Path homology of a cycle digraph \cite{grigor2012homologies} ) \label{thm:2}
    Let $G$ be a cycle digraph. Then, $\operatorname{dim} H_p(G)=0 \text { for all } p \geq 2$. If $G$ is a bigon or a (boundary) triangle or a (boundary) square, then $ \operatorname{dim} H_1(G)=0$; otherwise, $\operatorname{dim} H_1(G)=1$.
\end{theorem}
\begin{figure}[!t]
    \centering
    \includegraphics[width=1\linewidth]{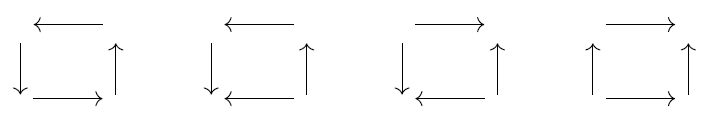}
    \caption{Four kinds of squares, where only the rightmost one has a trivial one-dimensional homology group. We call the rightmost one boundary square and call the others non-boundary squares.}
    \label{fig:squares}
\end{figure}

Consider a cycle digraph with only four vertices and four directed edges. Regardless of rotations and symmetries of the digraphs, which can be seen as having the same shape (for example, the following three digraphs: \begin{tikzcd}
    {} \arrow[r]           & {}           & {}           & {} \arrow[l]           & {} \arrow[d] & {} \arrow[l] \arrow[d] \\
    {} \arrow[u] \arrow[r] & {} \arrow[u] & {} \arrow[u] & {} \arrow[u] \arrow[l] & {}           & {} \arrow[l]          
\end{tikzcd}), we focus solely on the directions of the edges. Consequently, such digraphs only have four distinct shapes, as depicted in Fig. \ref{fig:squares}. According to Theorem \ref{thm:2}, only the rightmost shape exhibits a trivial one-dimensional homology group, leading us to refer to it as the \textit{boundary square}, whereas the remaining shapes are referred to as \textit{non-boundary squares} \cite{grigor2012homologies,dey2022efficient}. Throughout the subsequent discussion, when referring to the shape of a square (whether it is a boundary or non-boundary square), we disregard rotations and symmetries of the digraphs. In addition, when discussing a square, we consider it a digraph embedded in $\mathbb{R}^2$. Thus, we are also interested in the interior area of the square.

Now, we focus on another type of cycle digraph, namely, the \textit{polygon}, which consists of more than four vertices. Theorem \ref{thm:2} implies that all polygons possess non-trivial one-dimensional path homology groups, and the polygons themselves can serve as representations of the generators of these groups. Similar to squares, polygons are considered a digraph embedded in $\mathbb{R}^2$, whose interior area is also considered.

The one-dimensional path homology of simple digraphs is clear because of the following Theorem \ref{thm:3}. This theorem establishes the theoretical foundation for our subsequent exploration of the path homology of an angle-based grid digraph, which comprises boundary squares and non-boundary squares.
\begin{theorem}\label{thm:3} \cite{dey2022efficient}
    Let $G=(V, E)$ be a simple digraph. Let $Z_1=\operatorname{Ker}\left(\left.\partial\right|_{\Omega_1}\right), B_1=\operatorname{Im}\left(\left.\partial\right|_{\Omega_{2}}\right) $ , and let $Q$ denote the space generated by all bigons and boundary triangles and boundary squares in $G$. Then, we have $B_1=Q$. Hence the one-dimensional path homology group satisfies that $H_1=Z_1 /Q$. 
\end{theorem}
Hence, the following types of digraphs are all generators of $H_1$: 
$$\begin{tikzcd}
a \arrow[rd]          &             & a \arrow[r]           & b \arrow[d]           \\
c \arrow[u]           & b \arrow[l] & c \arrow[u]           & d \arrow[l]           \\
a \arrow[r]           & b \arrow[d] & a                     & b \arrow[l] \arrow[d] \\
c \arrow[u] \arrow[r] & d           & c \arrow[u] \arrow[r] & d                    
\end{tikzcd}$$

\section{Extraction of singular patterns from vector fields}\label{sec:4}
In this section, we develop an algorithm for extracting singularities. We first introduce the angle-based grid digraph, followed by the main algorithm and the analysis of its complexity. Finally, the approach for measuring variation between vector fields is introduced.

The process of the developed method unfolds as follows: Initially, we choose a specific area and extract the corresponding vector field data. Then we construct an angle-based grid digraph, where each edge’s weight is recorded. Subsequently, a digraph filtration is constructed by initially including all vertices of $G$ and subsequently adding edges of $G$ in a non-decreasing order of weights. This process enables the computation of the one-dimensional persistence diagram, thereby enabling the determination of the singularity’s position, as outlined in Algorithm \ref{alg:1}. If the vector field is time-varying, the changes of singular patterns can also be identified. The overall schematic of our methods is illustrated in Fig. \ref{fig:pipline}. This systematic approach ensures a comprehensive analysis of the vector field data and facilitates the identification and characterization of singular patterns.

In the following discussion, we assume that the following conditions hold:\\
(A1) In an angle-based grid digraph, any two vectors that are located on two adjacent points are not parallel. This conforms to most of the cases in practice.\\
(A2) The singularity does not appear on an edge.\\
(A3) The grid points are dense enough such that each smallest square encloses at most one singularity.

\subsection{Constructing the angle-based grid digraph}

We now introduce the primary tool used in this study, known as the \textit{angle-based grid digraph}. As illustrated in Fig. \ref{fig:vecrot}, consider a given vector field $X$. Let $\textbf{\textit{v}}_A=X(A)$ and $\textbf{\textit{v}}_B=X(B)$ represent the vectors where $X$ associates with points $A$ and $B$, respectively. When we mention that a field vector $\textbf{\textit{v}}_A$ rotates along $AB$ to reach the field vector $\textbf{\textit{v}}_B$, we refer to the vector moving from point $A$ to point $B$ while concurrently undergoing rotation. This rotation occurs within the range of angles less than $\pi$. Specifically, if the rotation angle is less than $\pi$ in the counterclockwise direction, we classify it as a counterclockwise rotation. Conversely, if the rotation angle exceeds $\pi$ in the counterclockwise direction, the rotation is considered clockwise. Note that under assumption (A1), the angle magnitudes are not equal to $\pi$.
\begin{figure}[!t]
    \centering
    \includegraphics[width=0.8\linewidth]{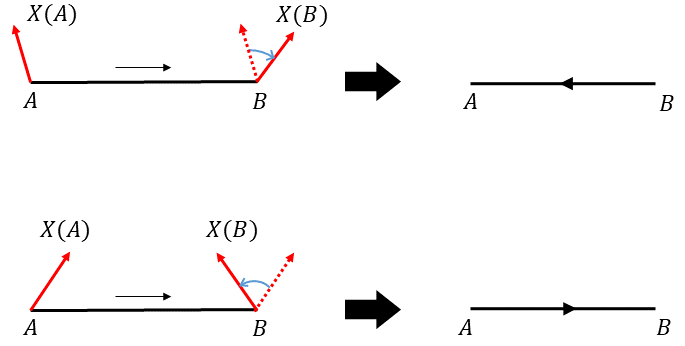}
    \caption{(Top) The field vector rotates clockwise when moving along edge $AB$; thus, the direction of $AB$ is $ B \to A$. (Bottom) The field vector rotates counterclockwise when moving along edge $AB$; thus, the direction of $AB$ is $ A\to B$.}
    \label{fig:vecrot}
\end{figure}

Let $X$ be a discrete vector field defined on $D \subseteq \mathbb{R}^2$, sampled at uniform grid points where the distance between adjacent points in the horizontal or vertical direction is a constant $\varepsilon >0$. We define $G(X,\varepsilon)=(V,E)$ as an \textit{angle-based $\varepsilon$-grid digraph} created by $X$. When the parameter $\varepsilon$ is clear from the context and does not cause confusion, we usually omit it and refer to $G(X)$ as the \textit{grid digraph} in short. The vertex set $V$ consists of all the sampling points of $X$ on $D$, and $E$ is the set of directed edges determined by the following rules (Fig. \ref{fig:vecrot} and \ref{fig:abd} ): for any two adjacent vertices $A$ and $B$ in $V$ (i.e. the distance between $A$ and $B$ is exactly $\varepsilon$), if the field vector $\textbf{\textit{v}}_A$ rotates counterclockwise along the line segment $AB$ to reach the field vector $\textbf{\textit{v}}_B$, then the directed edge is defined as $A \to B$. Conversely, if the field vector $\textbf{\textit{v}}_A$ rotates clockwise along the line segment $AB$ to reach the field vector $\textbf{\textit{v}}_B$, then the directed edge is defined as $B \to A$. The weight of a directed edge $AB$ is the absolute value of the rotation angle between the two vectors $\textbf{\textit{v}}_A$ and $\textbf{\textit{v}}_B$. Additionally, no directed edges are found between non-adjacent vertices. Fig. \ref{fig:abd} provides an example of an angle-based grid digraph created by a vector field.

According to assumption (A1), in the grid digraph $G$, no edge vanishing will occur between any two adjacent points. Thus, $G$ only consist of boundary squares and non-boundary squares, and we call these squares the \textit{minimal cells}.
According to Theorem \ref{thm:3}, the boundary squares in $G$ generate $B_1(G)$. Consequently, the minimal generators of the one-dimensional path homology $H_1(G)$ can be selected as all non-boundary squares in $G$.

\begin{figure}[!t]
    \centering
    \includegraphics[width=1\linewidth]{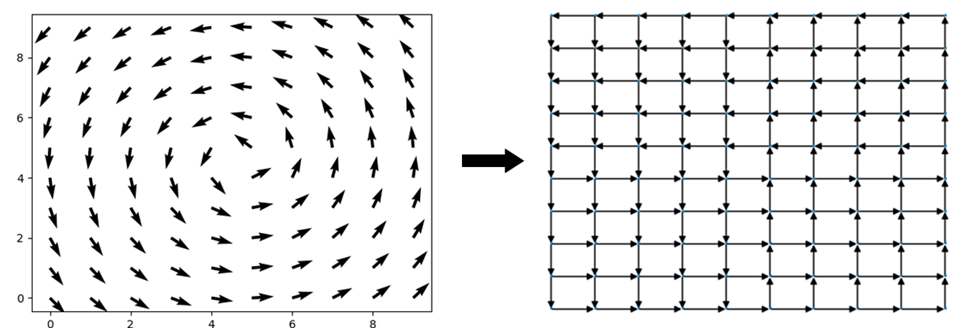}
    \caption{Example of angle-based grid digraph created by a vector field. The vertices of the digraph are the exact points where vectors are sampled.}
    \label{fig:abd}
\end{figure}

\subsection{Determining the position of singularities}
After obtaining the grid digraph $G(X)$ created by the planar vector field $X$, our task can be simplified by identifying the squares that encompass the singularities of $X$. In this manner, we can limit our search to those specific squares to identify all the singularities through the computation of one-dimensional persistent path homology.

Let us begin by considering the closed piecewise linear curve $L$ represented by the diagram:\begin{tikzcd}
    D \arrow[r, no head]                    & C                    \\
    A \arrow[u, no head] \arrow[r, no head] & B \arrow[u, no head]
\end{tikzcd}.
We assume that $L$ does not contain any singularities, and it can be denoted as $A\to B \to C \to D \to A$. The piecewise linear interpolation can be used to simplify the computation of the index of $L$ if the vector field $X$ is continuous \cite{tricoche2001continuous}. Let $\textbf{\textit{v}}_1,\textbf{\textit{v}}_2,\textbf{\textit{v}}_3,\textbf{\textit{v}}_4$ represent the four vectors that $X$ restricts on the four vertices $A,B,C,D$, respectively. The index of $L$ is then given by the following equation:
\begin{equation}\label{eqn:index}
    \operatorname{index}(L)=\frac{1}{2\pi} (\angle(\textbf{\textit{v}}_1,\textbf{\textit{v}}_2)+\angle(\textbf{\textit{v}}_2,\textbf{\textit{v}}_3)+\angle(\textbf{\textit{v}}_3,\textbf{\textit{v}}_4)+\angle(\textbf{\textit{v}}_4,\textbf{\textit{v}}_1))
\end{equation}
Here, $\angle(\textbf{\textit{v}}_i,\textbf{\textit{v}}_j)$ denotes the rotation angle from $\textbf{\textit{v}}_i$ to $\textbf{\textit{v}}_j$, which lies in the range $(-\pi,0)\cup (0,\pi)$. We define counterclockwise
rotation angles as positive and clockwise rotation angles as negative, and the sampling points of the discrete vector field are assumed to be sufficiently dense, ensuring that the absolute values of the angles are less than $\pi$ and not equal to zero.

Based on above discussions, the following theorem tells us that singularities can only exist in the interior area of non-boundary squares.
The proof of this theorem is shown in Appendix A.
\begin{theorem}\label{thm:6}
    Let $G$ represent the angle-based grid digraph created by the vector field $X$, with each smallest square in $G$ containing at most one singularity of $X$. We define $S$ as the set of all smallest squares that encompass a singularity of $X$ and $T$ as the set of representations of minimum generators of the one-dimensional path homology group of $G$. In other words, $T$ consists of all squares that have one of the following shapes:\begin{tikzcd}
        {} \arrow[d] & {} \arrow[l] & {} \arrow[d] & {} \arrow[l]           & {} \arrow[d] \arrow[r] & {}                     \\
        {} \arrow[r] & {} \arrow[u] & {}           & {} \arrow[u] \arrow[l] & {}                     & {} \arrow[u] \arrow[l]
    \end{tikzcd}. It follows that $S \subseteq T$.
\end{theorem}

We can also prove that singularities of $X$ could not exist in the interior of the squares that have the shape \begin{tikzcd}
    {}                     & {} \arrow[l] \arrow[d] \\
    {} \arrow[r] \arrow[u] & {}                    
\end{tikzcd}. Thus, the singularities of $X$ can only exist in the interior of the squares that have one of the following shapes:\begin{tikzcd}
    {} \arrow[d] & {} \arrow[l] & {} \arrow[d] & {} \arrow[l]           \\
    {} \arrow[r] & {} \arrow[u] & {}           & {} \arrow[u] \arrow[l]
\end{tikzcd}.

To extract the positions of singularities using one-dimensional persistent path homology, certain properties of singular patterns within a planar vector field $X$ must be considered. If a singular pattern exists in $X$, it must satisfy the following properties. The proofs of these properties are shown in Appendix A.

\begin{property}\label{prop:1}
    Every smallest square containing a singularity must have an edge $e$ with locally greatest weight, which means the weights of all edges (except $e$) of all squares containing $e$ are less than the weight of $e$.
\end{property}

\begin{property}\label{prop:2}
    Suppose $G$ is a grid digraph generated by $X$. We compute the one-dimensional persistent path homology of the digraph filtration created by $G$, as defined in Definition \ref{def:PPH}. This computation yields a corresponding persistence diagram. If the selected region has only one singularity, then after adding the edge with the greatest weight, a new representation of the minimum generator of $H_1(G)$ containing the singularity emerges.
\end{property}

In this property, the last square containing the singularity is formed by the edge with the greatest weight. Let the weight of this edge be denoted as $a$. This edge does not disappear in the digraph filtration, and therefore, it corresponds to a point $(a,+\infty)$ on the persistence diagram $Dgm_1$. Consequently, we can identify all the points on $Dgm_1$ that take the form $(a_i,+\infty)$, indicating the points that persist until infinity. For each $a_i$ and every edge in $G(X)$ with a weight of $a_i$, we examine the winding number of the adjacent smallest squares associated with this edge. The square with non-zero index is what we are looking for.

Once the smallest square containing the singularity has been identified, we can determine the approximate position of the singularity based on the weights of the edges. While multiple approaches may be used to obtain the approximate coordinates, all edges evidently have the same weight, and the singularity should be located at the geometric center of the square. However, given that edges closer to the singularity possess greater weight, we can estimate the singularity’s position using a weighting method. In this study, we utilize the concept of \textit{the weighted center of the square} to approximate the singularity's coordinates, given by:
\begin{equation}\label{eq2}
    (x,y)=\left(\sum_{i=1}^{4} \frac{w_i x_i}{w},\sum_{i=1}^{4} \frac{w_i y_i}{w}\right)
\end{equation}
Here, $(x_i,y_i)$ represents the midpoint coordinates of the four edges of the square, with corresponding weights $w_i$ for $i=1,2,3,4$. The total weight is denoted as $w$, calculated as $w=\sum_{i=1}^4 w_i$. 
Hence, the center can only be in the diamond formed by the edge midpoints of the square's edges.
When considering two parallel edges within the square, the weight center is closer to the edge with a greater weight, similar to the singularity. Consequently, horizontally and vertically, the weight center and the singularity are positioned on the same side of the geometric center of the square. Therefore, if we use the notation $\varepsilon$ to denote the length of the square's edges, the errors in the $x$ and $y$ coordinates between the weight center and the singularity are less than $\varepsilon /2$.
The procedure for extracting the position of singularity is included in Algorithm \ref{alg:1}, illustrated in subsection \ref{D}.

\subsection{Algorithm and complexity analysis}\label{D}
Our algorithm for extracting singularities from a planar discrete vector field can be described as the Algorithm \ref{alg:1}. 
To analyze its complexity, we assume the input $X$ has a rectangular grid of points with $m$ points in each row and $n$ points in each column (hence $mn$ grid points in $X$ in total), and in $Dgm_1(G)$ there are $k$ points persist to $+\infty$. Note that $k$ represents the number of non-boundary squares in $G$, thus is not greater than $mn$. We consider the time complexity of each steps. Step 1 has a complexity of $O(mn)$, which represents a linear relationship between the input sizes $m$ and $n$. The computation persistent path homology of planar digraph in step 2 exhibits a complexity of $O((mn)^\omega)$ \cite{dey2022efficient}, where $\omega$ denotes the exponent of matrix multiplication, currently known to be less than $2.373$. This step is particularly significant due to the high exponent, which can dominate the overall complexity for large input sizes. The sort process in step 3 has a complexity of $O(k \log k)$. Step 4 has a complexity of $O(kmn)$, but it does not exceed $O((mn)^\omega)$ since $k$ is controlled by $mn$ and is much smaller than $mn$ in general cases.
Considering the relative magnitudes of these complexities, the step of computing persistent path homology with $O((mn)^\omega)$ stands out as the dominant term, determining the overall complexity of the algorithm. Therefore, if there are $N$ grid points in $X$, the overall time complexity of Algorithm \ref{alg:1} is $O((N)^\omega)$.

\begin{algorithm}[H]
    \renewcommand{\algorithmicrequire}{\textbf{Input:}}
    \renewcommand{\algorithmicensure}{\textbf{Return:}}
    \caption{Extracting Singularities of a Vector Field}
    \label{alg:1}
    \begin{algorithmic}[1]
        \REQUIRE A planar discrete vector field $X$ containing singularities.
        \STATE Compute the grid digraph $G$ from $X$.
        \STATE Compute the persistence diagram $Dgm_1(G)$ using the one-dimensional persistent path homology of the digraph filtration created by $G$.
        \STATE Let $I$ be the set of points in $Dgm_1(G)$ of the form $(a_i,+\infty)$, and sort the elements in $ I $ in non-decreasing order of $a_i$.
        \STATE \textbf{For} each point $(a_i,+\infty)$ in $I$ \textbf{do}\\
        \hspace{1em} \textbf{For} each edge with weight $a_i$ in $G$ \textbf{do}\\
        \hspace{2em} \textbf{For} each adjacent smallest square to this edge \textbf{do}\\
        \hspace{3em} Calculate the winding number of the square.\\
        \hspace{3em} \textbf{If} the winding number is 1 \textbf{then}\\
        \hspace{4em} Compute the singularity $(x,y)$ using Eq. (\ref{eq2}).\\
        \hspace{3em} \textbf{End if}\\
        \hspace{2em} \textbf{End for}\\
        \hspace{1em} \textbf{End for}\\
        \textbf{End for}
        \ENSURE The singularities of the vector field $X$. 
    \end{algorithmic}
\end{algorithm}

\section{Differences Measurement of Vector Fields}\label{sec:5}
Our method also introduces a novel approach to compare the topology of two planar vector fields using persistent path homology theory. This involves computing the distances of 1-dimensional persistence diagrams obtained from two vector fields. Assuming two planar vector fields $X_1$ and $X_2$ are defined on the same domain, computing their respective 1-dimensional persistence diagrams $PD_1$ and $PD_2$, and subsequently evaluating the corresponding distance $d(PD_1,PD_2)$ (here $d$ represents the bottleneck distance or Wasserstein distance), allows us to measure the dissimilarity between the two persistence diagrams, thereby comparing the topological differences between $X_1$ and $X_2$. 

Moreover, when considering a time-varying vector field restricted to a fixed domain, we can track the variations in topology by sequentially computing the distance between two persistence diagrams obtained at two neighboring moments. A small distance between the persistence diagrams of neighboring moments suggests a minor variation in the singular pattern of the vector field, indicating greater stability during this period. Conversely, a large distance indicates a significant variation in the singular pattern of the vector field, indicating less stability during this period. This will be helpful to quantify and analyze changes in the vector field.

In practice, as we concentrate on the relative magnitude of the distances between persistence diagrams, using either the bottleneck distance or the Wasserstein distance can have a similar effect, given that the bottleneck distance can be considered as a special case of the Wasserstein distance. Therefore, in the subsequent experiments, we employ the bottleneck distance as the primary measuring tool.

\begin{figure*}[!t]
    \centering
    \includegraphics[width=1.0\linewidth]{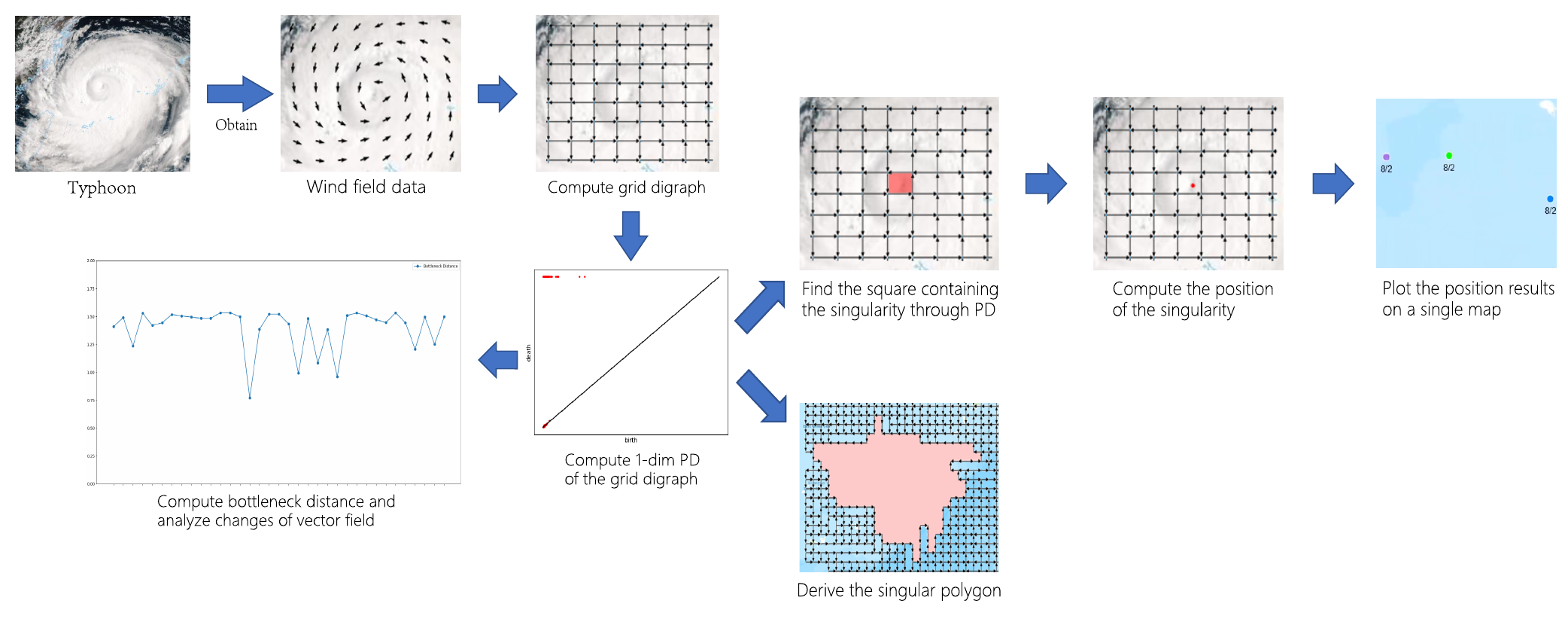}
    \caption{Pipeline of our method for analyzing the singular patterns of vector fields.}
    \label{fig:pipline}
\end{figure*}

\begin{figure*}[!t]
    \centering
    \includegraphics[width=0.9\linewidth]{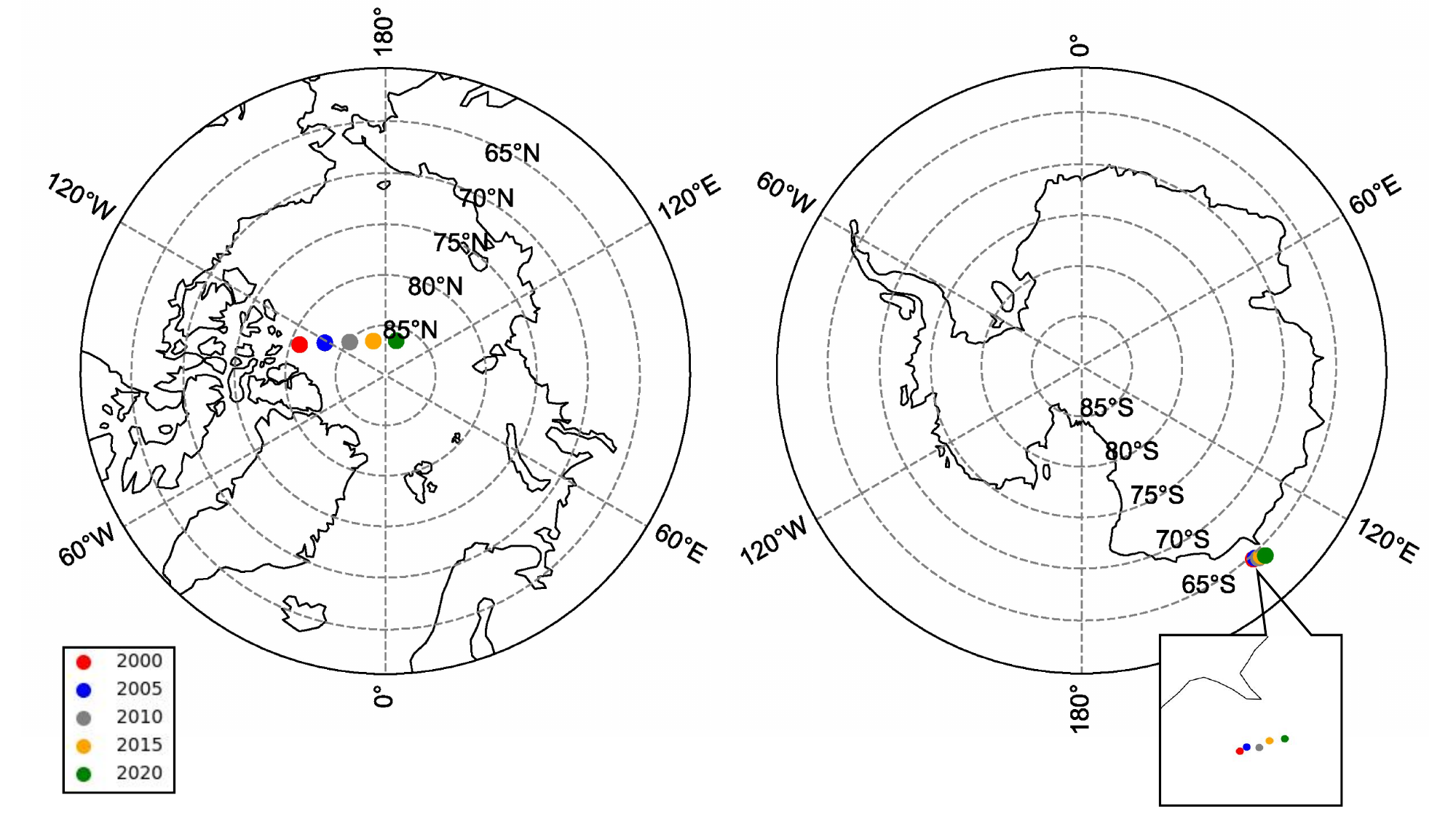}
    \caption{The positioning results of north dip pole (left) and south dip pole (right) for the years 2000, 2005, 2010, 2015, and 2020.}
    \label{fig:magf}
\end{figure*}
\begin{figure}[!t]
    \centering
    \subfloat[]{\includegraphics[width=1\linewidth]{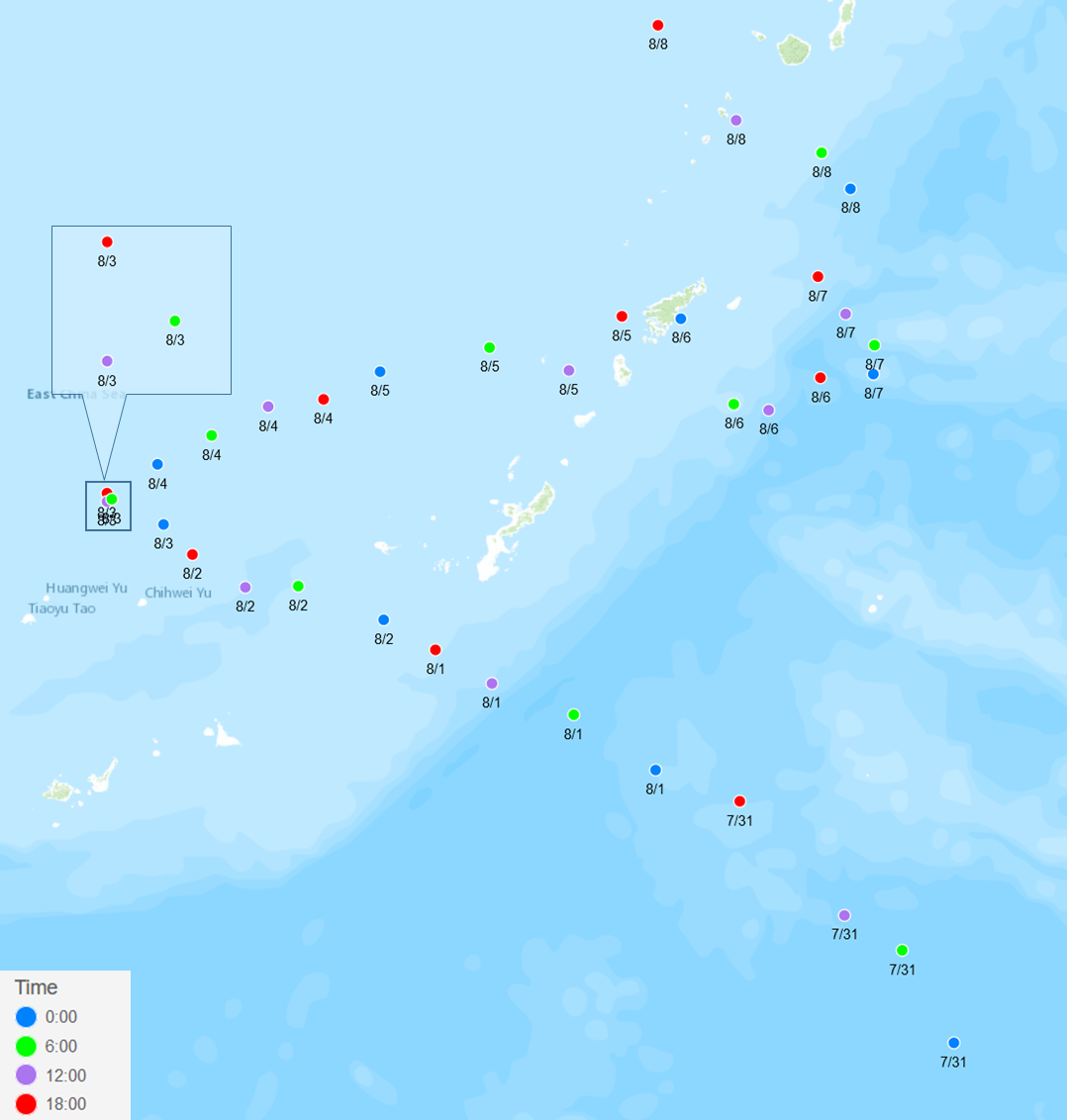}%
        \label{fig:kn}}
    \hfil
    \subfloat[]{\includegraphics[width=1\linewidth]{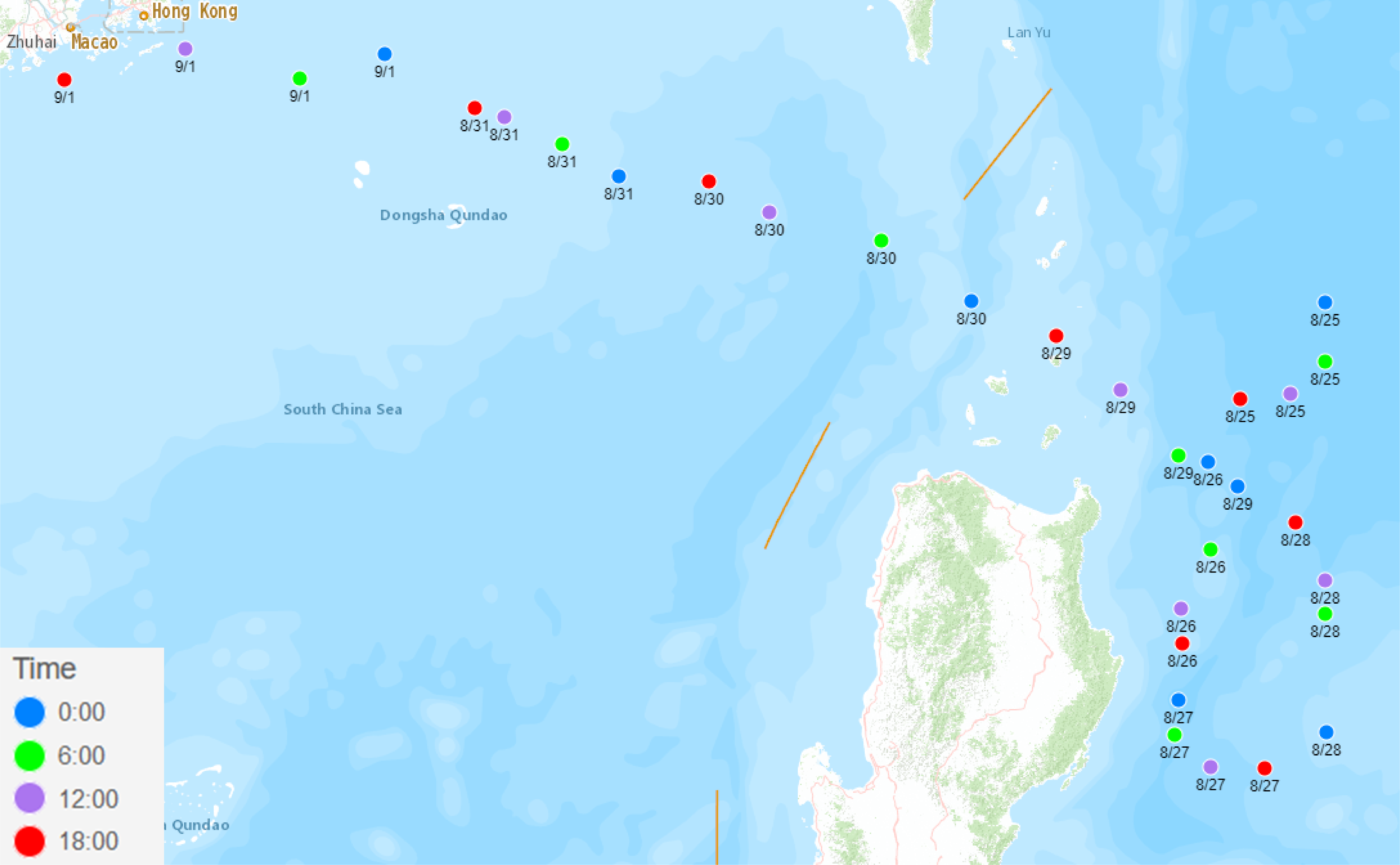}%
        \label{fig:sl}}
    \caption{Center tracking results of two typhoons. (a) Results of typhoon Khanun from 2023/07/31 to 2023/08/08. (b) Results of typhoon Saola from 2023/08/25 to 2023/09/01.}
    \label{fig:typhoon}
\end{figure}

\section{Applications and discussions}\label{sec:6}
All the experiments in this section were conducted on a server equipped with a 2.90GHz Intel Core i7-10700 CPU and 16.0 GB RAM. Additionally, Fig. \ref{fig:pipline} shows the pipeline of our procedure, and the algorithm proposed in \cite{dey2022efficient} is used for computing one-dimensional persistent path homology.

\subsection{Positioning magnetic dip poles}
To begin with, we present the application of our method to determine the position of Earth's magnetic dip poles using the geomagnetic field datasets obtained from the International Geomagnetic Reference Field (IGRF) \cite{alken2021international}, a mathematical model describing the geomagnetic field and its changes. The magnetic dip poles are defined as the locations where the main magnetic field as a whole is normal to Earth’s surface, which provides a crucial reference for local orientation when navigating on or close to Earth’s surface at high-latitudes \cite{alken2021international}. However, observations have shown that the positions of the two magnetic dip poles change from year to year \cite{alken2021international}, so it is necessary to calculate and update their positions in time. Our approach is well-suited to address this issue. When we project geomagnetic field data onto a map, resulting in a planar vector field, the two dip poles (north dip pole and south dip pole) become singularities of this planar vector field, since the projected vectors are zero at these points. In this experiment, we select two domains around the two dip poles respectively and choose geomagnetic field data at zero altitude. After projecting onto the map, discrete vector fields sampled around the two dip poles are obtained. Using our method, we determine the locations of the two dip poles for the years 2000, 2005, 2010, 2015, and 2020, and compare our results with the data of the actual location provided by IGRF \cite{alken2021international}. These results are available in the Appendix B (Table. \ref{tab:IGRF}). It is evident that our method incurs only minor positioning errors, thus reaffirming its effectiveness. The positioning results are also plotted on maps (Fig. \ref{fig:magf}).

\subsection{Tracking the center of typhoons} 
Another application is tracking the center of two tropical cyclones, namely Khanun and Saola, which emerged in 2023 and exhibited unusual trajectories. Here we use wind field datasets obtained from the Remote Sensing Systems Cross-Calibrated Multi-Platform's 6-hourly ocean vector wind analysis product on a 0.25 deg grid \cite{mears2022rss}; it is a gridded Level 4 (L4) product that provides vector wind over the world's oceans and provides gridded wind vector data at 4 times of day: 00:00Z, 06:00Z, 12:00Z, and 18:00Z (here Z stands for Universal Time Coordinated, UTC). 
Based on these wind fields, we determine the centers of tropical cyclones (which can be regarded as spiral sinks). By processing the wind field data, we acquire the approximate latitude and longitude of the center positions of the tropical cyclone for each provided moment. These positions are then plotted on a single map, as depicted in Fig. \ref{fig:typhoon} (a) displays the positioning result for Khanun, whereas Fig. \ref{fig:typhoon} (b) presents the positioning result for Saola. Each point on the map is labeled with the corresponding date, and the color indicates the observation time. 
The obtained positioning longitudes and latitudes for the centers of Saola and Khanun are provided in Appendix B (Table. \ref{tab:SL} and Table. \ref{tab:KN}), alongside the actual longitudes and latitudes, as well as the errors. 
Here, the actual positions of the tropical cyclones are obtained from the Regional Specialized Meteorological Center for Marine Meteorological Service, Beijing\footnote{http://eng.nmc.cn/typhoon/}.

\subsection{Vector field difference measurement}
In the following, we will use the wind field datasets that were used in the previous subsection as an example to demonstrate how our method can measure the variation of time-varying vector fields. Similar to the work \cite{Soler2018LiftedWM}, we will use the resulting persistence diagrams obtained by wind fields to measure the variations of singular patterns of the wind fields of the two tropical cyclones, respectively. 

Specifically, we sequentially compute the bottleneck distance between the persistence diagrams of two neighboring moments. The results are displayed in Fig. \ref{fig:bottleneck}, where the point labeled $k$ represents the bottleneck distance between the 1-dimensional persistence diagram at the $(k+1)$-th moment and the persistence diagram at the $k$-th moment. For instance, the first moment of Khanun is 2023-07-31-06:00 UTC; thus, the first blue point's value represents the bottleneck distance between the persistence diagram obtained at 2023-07-31-06:00 UTC and the one obtained at 2023-07-31-00:00 UTC. The first moment of Saola is 2023-08-25-06:00 UTC; hence, the first orange point's value represents the bottleneck distance between the persistence diagram obtained at 2023-08-25-06:00 UTC and the one obtained at 2023-08-25-00:00 UTC. 

From the blue line in Fig. \ref{fig:bottleneck}, it is evident that the bottleneck distance reached a local minimum at the moment 2023-08-03-18:00 UTC, coinciding with the time when Khanun became stationary near a position since a notable change in trajectory was taking place. This period exhibited reduced variability in the wind field, resulting in a smaller difference between the persistence diagrams.
Similarly, in the orange line in Fig. \ref{fig:bottleneck}, at the moment 2023-08-27-18:00 UTC, Saola became stationary near a position where its trajectory was changing direction, indicating small variability between the persistence diagrams of the corresponding wind fields. Similar features of bottleneck distance can be observed at other moments, such as 2023-08-31-06:00 UTC, during which period Saola moves slowly within the vicinity of a location.
Conversely, if a tropical cyclone undergoes significant movement or multiple turns over a period of time, the corresponding bottleneck distance will be larger. Moreover, a significant fluctuation of the bottleneck distance may indicate that the tropical cyclone is undergoing unstable movement and change.

Hence, considering our finding that the magnitude of the bottleneck distance between the corresponding persistence diagrams of the wind fields over a specific period reflects the extent of changes in the singular patterns of the tropical cyclones during that time, coupled with the widespread belief that these changes may result from the influence of various natural factors on the tropical cyclones, we also believe that our method can be applied to analyze the factors contributing to the variations in tropical cyclones.

\begin{figure*}[!t]
    \centering
    \includegraphics[width=0.65\linewidth]{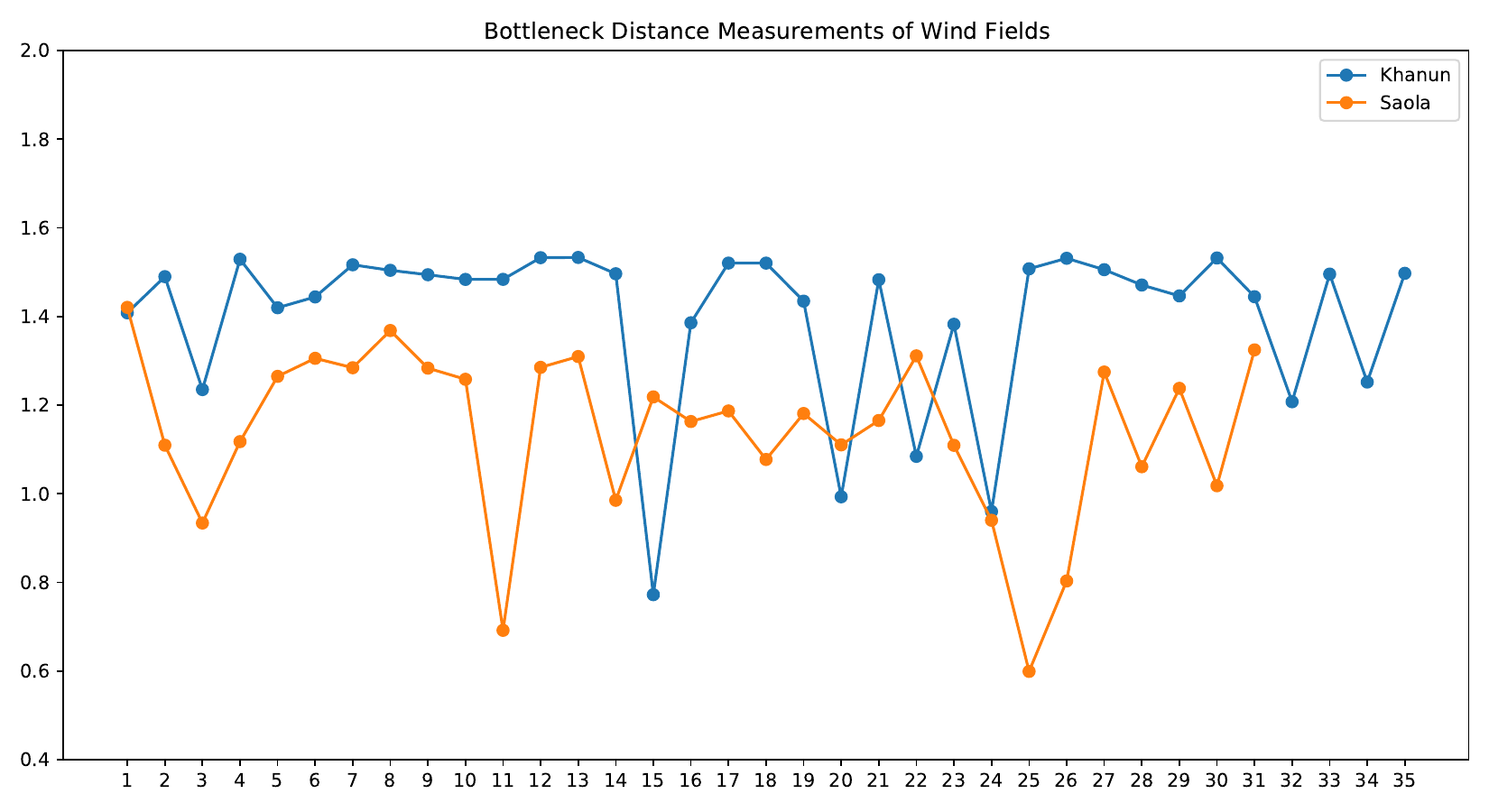}
    \caption{Bottleneck distances computed by one-dimensional persistence diagrams derived from Khanun and Saola datasets respectively.}
    \label{fig:bottleneck}
\end{figure*}
\begin{figure*}[!t]
    \centering
    \includegraphics[width=1.0\linewidth]{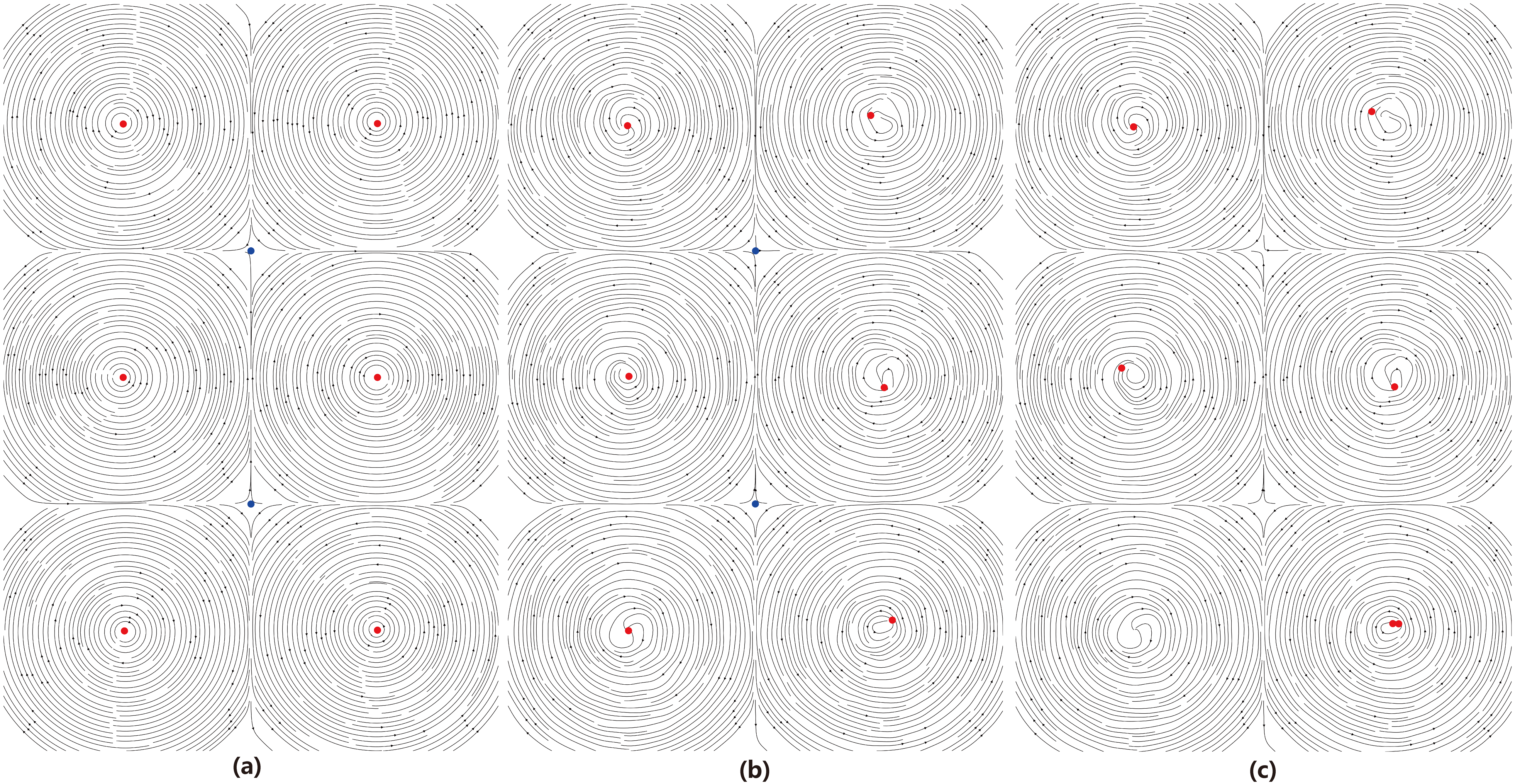}
    \caption{Robustness of our method to noise. (a) A vector field is presented by its streamlines, and six center points (red) and two saddles (blue) are successfully extracted by our method. (b) Noise is added in the vector field in (a), while our method can keep identifying the singularities successfully. (c) Comparison of singularities identified by triangular linear interpolation method \cite{tricoche2002topology} from (b). It fails to detect some singularities while incorrectly identifying a spurious singularity.}
    \label{fig:noise}
\end{figure*}
\begin{figure*}[!t]
    \centering
    \includegraphics[width=1.0\linewidth]{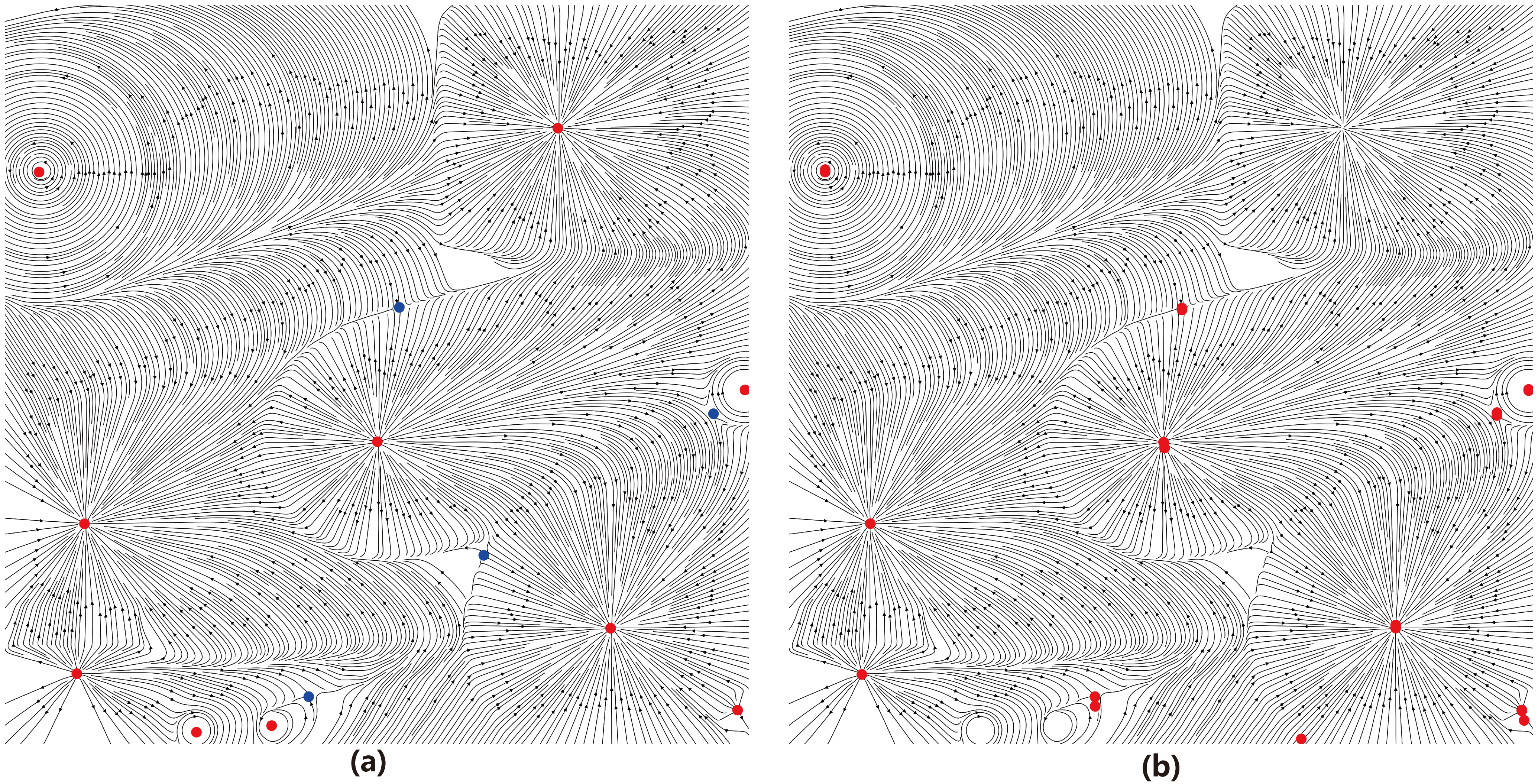}
    \caption{The streamlines of a vector field with 10,000 vectors. (a) Singularities identified by our method. The detected singularities with positive index are drawn in red, while singularities with negative index are drawn in blue. (b) Comparison of singularities identified by triangular linear interpolation method \cite{tricoche2002topology}. It fails to detect some singularities while incorrectly identifying some spurious singularities.}
    \label{fig:largescale}
\end{figure*}

\subsection{Robustness and comparisons}
In this section, we first demonstrate the robustness of our method for detecting singular patterns against noise, using the vector field shown in Fig. \ref{fig:noise} (a), which contains six center points and two saddles. Here, to show the increasing noise, we draw these vector fields in the form of streamlines. As noise is added in the vector field (Fig. \ref{fig:noise} (b)), our method can still identify all singularities without omission or false detection. In contrast, the triangular linear interpolation method \cite{tricoche2002topology} (Fig. \ref{fig:noise}(c)) fails to detect one center point and two saddle points while incorrectly identifying a spurious singularity (such a point has index 0 but is identified as a singularity). 
Moreover, the proposed method shows robustness in analyzing singular patterns in large-scale vector fields. As demonstrated in Fig. \ref{fig:largescale} (a), which shows streamlines of a noisy vector field comprising 10,000 vectors with multiple singularity types, our approach successfully identifies all singularities and correctly classifies their types, while in Fig. \ref{fig:largescale} (b), the triangular linear interpolation method does not find all singularities and identify their types correctly. These results clearly demonstrate the robustness of our approach for handling noise and large-scale vector fields.

Then we compare the mean positioning errors and running time between our method and several classical methods.
To calculate the mean errors, suppose that we have data for $n$ moments, and we denote the $i$-th actual longitude as $al_i$ and the $i$-th positioning longitude as $pl_i$. The mean longitude error can be computed using the formula $\frac{1}{n}\sum_{i=1}^{n} \left|al_i-pl_i \right|$, and the mean latitude error can be computed similarly. Table \ref{tab:comknsl} provides a comparison of the mean errors obtained using our method with other singularity extraction methods, they are methods based on triangular linear interpolation of vectors \cite{tricoche2002topology}, computation of the Jacobian matrix \cite{helman1989representation}, and Hodge decomposition \cite{polthier2003identifying} respectively. 
The average running times are also shown in Table \ref{tab:comknsl}. The Khanun datasets consists of 1600 sampling vectors (resulting in a digraph of 1600 vertices and 3120 edges), the Saola datasets comprises 1924 sampling vectors (resulting in a digraph of 1924 vertices and 3759 edges), and the geomagnetic field datasets consists of 1600 sampling vectors (resulting in a digraph of 1600 vertices and 3120 edges). 
Based on the results, our method achieves smaller errors than the method proposed by \cite{polthier2003identifying}, and our time cost is lower than \cite{polthier2003identifying}, which involves performing Hodge decomposition and locating local extrema of two energy functionals derived from the rotation-free and divergence-free components on the refined grids.
Furthermore, our approach can achieve errors nearly as small as the numerical methods of \cite{tricoche2002topology} and \cite{helman1989representation}, although it has a slightly longer time cost than those numerical methods. However, these numerical methods may fail to correctly identify singularities or encounter spurious singularities due to precision issues.
For instance, in the experiments, both the triangular linear interpolation method and the Jacobian matrix method yield spurious singularities. The number of spurious singularities obtained from all vector fields in the corresponding datasets is shown in Table \ref{tab:comknsl}. For example, in the geomagnetic field datasets, there are 10 vector fields and 10 corresponding singularities in total, but the triangular interpolation method finds 15 singularities, including 5 spurious singularities, and the Jacobian method finds 12 singularities, including 2 spurious singularities. In contrast, our method, which relies on the singularity index, finds all of the singularities without any spurious ones. Additionally, our method has the capability to measure the variations between vector fields, while the other methods are unable to accomplish this.

\begin{table*}[!t]  
    \centering  
    \caption{Comparative Results}
    \label{tab:comknsl}  
    \resizebox{0.8\textwidth}{!}{%  
        \begin{tabular}{|c|c|c|c|c|}  
            \hline  
            \multicolumn{5}{|c|}{Comparative Results for Khanun Datasets (36 vector fields)} \\  
            \hline  
            Method & Errors of longitude($^{\circ}$) & Errors of latitude($^{\circ}$) & Running time(s) & Spurious singularity\\   
            \hline
            Triangular interpolation \cite{tricoche2002topology} & 0.134 & 0.107 & 27.21 &15\\  
            \hline
            Jacobian method \cite{helman1989representation} & 0.126 & 0.109 & 45.25 & 9\\ 
            \hline  
            Hodge decomposition \cite{polthier2003identifying} & 0.205  & 0.147 & 566.64 & 0\\  
            \hline
            \textbf{our method} & 0.131 & 0.118 & 70.19 & 0\\ 	  
            \hline  
            \multicolumn{5}{|c|}{Comparative Results for Saola Datasets (32 vector fields)} \\  
            \hline  
            Method & Errors of longitude($^{\circ}$) & Errors of latitude($^{\circ}$) & Running time(s) & Spurious singularity\\     
            \hline
            Triangular interpolation \cite{tricoche2002topology} & 0.108 & 0.093 & 34.76 & 15\\  
            \hline
            Jacobian method \cite{helman1989representation} & 0.111 & 0.093 & 56.00 &2\\  
            \hline  
            Hodge decomposition \cite{polthier2003identifying} & 0.129  & 0.133 & 689.39 & 0\\
            \hline
            \textbf{our method} & 0.110 & 0.102 & 103.08 &0\\  
            \hline  
            \multicolumn{5}{|c|}{Comparative Results for Geomagnetic Field Datasets (10 vector fields)} \\  
            \hline  
            Method & Errors of longitude($^{\circ}$) & Errors of latitude($^{\circ}$) & Running time(s) & Spurious singularity\\  
            \hline
            Triangular interpolation \cite{tricoche2002topology} & 0.028 & 0.056 & 27.90 &5\\  
            \hline
            Jacobian method \cite{helman1989representation} & 0.030 & 0.000 & 47.31 & 2\\  
            \hline  
            Hodge decomposition \cite{polthier2003identifying} & 0.089  & 0.182 & 572.02 &0\\ 
            \hline
            \textbf{our method} & 0.037 & 0.059 & 90.95 &0\\ 	  
            \hline  
        \end{tabular}%  
    }  
\end{table*}

\subsection{Discussion}
Our method exhibits several limitations. If the sampling grid is not sufficiently dense, some singularities may not be extracted because multiple singularities may exist within the same square, and their indices can cancel each other out. For example, if a singularity with index 1 and a singularity with index -1 locate in an identity square, the index of the square will finally becomes 0 and will not be identified as containing singularities.
Furthermore, singularities located on the edges of the grid digraph may be overlooked, since the direction of this edge is unsure and the position of the singularity on the edge cannot be determined accurately. 
But in general, these situations are rarely to occur due to practical considerations and limitations in numerical accuracy.
Additionally, since the generators of 2-dimensional path homology group is unclear currently, it is hard to generalize the proposed method to the vector fields in higher dimensional space.

\section{Conclusions}\label{sec:7}
In this study, we propose a novel method for determining the singularity within a planar vector field exhibiting a singular pattern. Moreover, we compare the changes of singular patterns of time-varying vector fields using persistence diagrams. Our approach involves the conversion of discrete vector fields into angle-based grid digraphs, which are then analyzed using persistent path homology. Experimental results demonstrate that our proposed method achieves low mean errors in locating the centers of tropical cyclones and locations of dip poles, and effectively measures the changes in tropical cyclone features over time. 

In future work, we aim to explore alternative techniques for converting vector fields into digraphs to accommodate diverse applications. 
Furthermore, we will study how to address challenges related to vector fields with high-order singularities and vector fields in higher dimensional Euclidean spaces or surfaces, and investigate the potential utilization of higher dimensional (persistent) path homology as a new analytical tool.

\section*{Acknowledgments}
This work is supported by the National Natural Science Foundation of China under Grant nos. 62272406, 61972316. CCMP Version-3.0 vector wind analyses are produced by Remote Sensing Systems. Data are available at www.remss.com.

\section*{CRediT authorship contribution statement}
\textbf{Yu Chen:} Conceptualization, Investigation, Formal analysis, Methodology, Software, Validation, Writing – original draft, Writing – review \& editing.
\textbf{Hongwei Lin:} Conceptualization, Formal analysis, Methodology, Supervision, Funding acquisition, Writing – original draft, Writing – review \& editing.

\section*{Declaration of competing interest}
The authors declare that they have no known competing financial interests or personal relationships that could have appeared to influence the work reported in this paper.

\section*{Data availability}
Data will be made available on request.

\bibliographystyle{cag-num-names}
\bibliography{bib1.bib}

\section*{Appendices}
\begin{appendices}
\section{Proofs of some theorems and properties}
\subsection{Proof of the Theorem 3}
\begin{proof}
	It is sufficient to prove that the singularities could not exist in the interior area of the boundary squares. Without loss of generality, let us consider the boundary square:
	\begin{tikzcd}
		C \arrow[r]           & D           \\
		A \arrow[r] \arrow[u] & B \arrow[u]
	\end{tikzcd}.
	According to the definition of the angle-based grid digraph, when moving from $A$ to $B$ and from $B$ to $D$, the vector undergoes a counterclockwise rotation, with the rotation angle ranging between $(0,\pi)$. Conversely, when moving from $D$ to $C$ and from $C$ to $A$, the vector rotates clockwise, with the rotation angle ranging between $(-\pi,0)$. By utilizing equation 
	$$\operatorname{index}(L)=\frac{1}{2\pi} (\angle(\textbf{\textit{v}}_1,\textbf{\textit{v}}_2)+\angle(\textbf{\textit{v}}_2,\textbf{\textit{v}}_3)+\angle(\textbf{\textit{v}}_3,\textbf{\textit{v}}_4)+\angle(\textbf{\textit{v}}_4,\textbf{\textit{v}}_1))$$
	we observe that along the closed curve $ABDCA$, the sum of rotation angles falls within the range of $(-2\pi,2\pi)$. Consequently, the number of revolutions made by the vector field $X$ while traveling around the curve $ABDCA$ is zero because zero is the only integer within the interval $(-1,1)$. As stated in Theorem 1 in main paper, this implies the absence of any singularities of $X$ within the interior of the boundary square.
\end{proof}
\begin{remark}
	In fact, for the square type
	\begin{tikzcd}
		C \arrow[d] \arrow[r] & D                     \\
		A                     & B \arrow[l] \arrow[u]
	\end{tikzcd}, its index is also zero since zero is the only integer within the interval $(-1,1)$. Thus this type of square will not contain singularity as well.
\end{remark}

\subsection{Proof of the Property 1 and Property 2}
We first prove the following Lemma:
\begin{lemma}\label{main}
	Let $X$ be a discrete vector field with a singularity $S$. It follows that a neighborhood $U$ of $S$ exists, where when moving in a consistent direction along a horizontal or vertical line within $U$ that does not intersect $S$, the field vector consistently rotates in the same direction. Consequently, let $G$ be the corresponding angle-based $\varepsilon$-grid digraph; then, in $G$, the edges situated on the same horizontal or vertical line have same direction. Hence, the singular patterns of the vector field can be preserved in the digraph during the transformation process.
\end{lemma}

\begin{proof}
	A number of singular patterns such as sources, sinks, circulating and spiraling vector fields can be modeled or approximated by a field generated by logarithmic spirals \cite{wong2009identifying}. The formula of each grid point can be discribed by
	$$\left\{\begin{matrix}
		x=a e^{\theta \cot\alpha } \cos\theta   \\
		y=\frac{a}{\rho} e^{\theta \cot\alpha } \sin\theta
	\end{matrix}\right.$$
	and the vector located on $(x,y)$ is
	$$\left\{\begin{matrix}
		P=a e^{\theta \cot\alpha } (\cot\alpha \cos\theta-\sin\theta )=x \cot\alpha-\rho y	\\
		Q=	\frac{a}{\rho} e^{\theta \cot\alpha }(\cot\alpha \sin\theta+\cos\theta)=y \cot\alpha +\frac{x}{\rho}
	\end{matrix}\right.$$
	where $\theta$ is the polar angle, $\alpha \in [0,\pi)$ is the angle between the radial line and the tangent of the spiral at $(r,\theta)$; parameter $a$ is non-zero, and a counterclockwise field is generated by a positive $ a $ while a clockwise field is generated by a negative $a$; parameter $\rho$ controls the width-to-height ratio.
	
	To prove this theorem, we only need to consider the case of a horizontal line, as the vertical case can be proven similarly. Without loss of generality, let us assume that the singularity is located at $O(0,0)$. Suppose $(x,y)$ and $(x+\varepsilon,y)$ are two points in $U$, which can be described in polar coordinates as follows:
	$$
	\begin{aligned}
		x &= a_1 e^{\theta_1 \cot\alpha_1} \cos\theta_1, \quad x+\varepsilon = b_1 e^{\theta_2 \cot\alpha_2} \cos\theta_2 \\
		y &= \frac{a_1}{\rho} e^{\theta_1 \cot\alpha_1} \sin\theta_1 = \frac{b_1}{\rho} e^{\theta_2 \cot\alpha_2} \sin\theta_2
	\end{aligned}
	$$
	Here, $a_1$, $b_1$, $\theta_1$, $\theta_2$, $\alpha_1$, $\alpha_2$ and $\rho$ are constants such that $y\ne 0$, with $a_1$ and $b_1$ both positive or negative. By rewriting the equation, we have:
	$$
	\begin{aligned}
		b_1 e^{\theta_2 \cot\alpha_2} &= b_1 e^{\theta_2 (\cot\alpha_2 - \cot\alpha_1 + \cot\alpha_1)} \\
		&= (b_1 e^{\theta_2 (\cot\alpha_2 - \cot\alpha_1)}) e^{\theta_2 \cot\alpha_1}
	\end{aligned}
	$$
	Let $a_2 = b_1 e^{\theta_2 (\cot\alpha_2 - \cot\alpha_1)}$ and $\alpha = \alpha_1$. Then, we can express $x$, $x+\varepsilon$, $y$ as:
	$$
	\begin{aligned}
		x &= a_1 e^{\theta_1 \cot\alpha} \cos\theta_1, \quad x+\varepsilon = a_2 e^{\theta_2 \cot\alpha} \cos\theta_2 \\
		y &= \frac{a_1}{\rho} e^{\theta_1 \cot\alpha} \sin\theta_1 = \frac{a_2}{\rho} e^{\theta_2 \cot\alpha} \sin\theta_2
	\end{aligned}
	$$
	Note that $a_2$ and $\alpha$ are also constants. Therefore, $\varepsilon$ can be written as:
	$$
	\begin{aligned}
		\varepsilon &= a_2 e^{\theta_2 \cot\alpha} \cos\theta_2 - a_1 e^{\theta_1 \cot\alpha} \cos\theta_1 \\
		&= a_2 e^{\theta_2 \cot\alpha} \cos\theta_2 - a_2 e^{\theta_2 \cot\alpha} \frac{\sin\theta_2}{\sin\theta_1}\cos\theta_1 \\
		&= a_2 e^{\theta_2 \cot\alpha} \frac{\sin(\theta_1-\theta_2)}{\sin\theta_1}
	\end{aligned}
	$$
	Now, let's define:
	$$
	\begin{aligned}
		X_1 &= (P(x,y),Q(x,y),0) \\
		&= (x \cot\alpha - \rho y, y \cot\alpha + \frac{x}{\rho}, 0)
	\end{aligned}
	$$
	$$
	\begin{aligned}
		X_2 &= (P(x+\varepsilon,y),Q(x+\varepsilon,y),0) \\
		&= ((x+\varepsilon) \cot\alpha - \rho y, y \cot\alpha + \frac{x+\varepsilon}{\rho}, 0)
	\end{aligned}
	$$
	Thus, the third component of $X_2 \times X_1$ is given by:
	$$\begin{aligned}  
		z=&((x+\varepsilon) \cot\alpha-\rho y)(y \cot\alpha +\frac{x}{\rho})\\  
		&-(x \cot\alpha-\rho y)(y \cot\alpha +\frac{x+\varepsilon}{\rho}) \\  
		=&\varepsilon y(1+\cot^2 \alpha) \\  
		=& a_2 y(1+\cot^2 \alpha) e^{\theta_2 \cot\alpha }\frac{\sin(\theta_1-\theta_2)}{\sin\theta_1}  
	\end{aligned}$$  
	If $y>0$, then $\theta_1$ and $\theta_2$ lie in the interval $(0,\pi)$, and $\sin(\theta_1-\theta_2)>0$ and $\sin\theta_1>0$ (see left part in Fig. \ref{fig:showline}), thus $z$ has the same sign as $a_2$. If $y<0$, then $\theta_1$ and $\theta_2$ lie in the interval $(\pi,2\pi)$, and $\sin(\theta_1-\theta_2)<0$ and $\sin\theta_1<0$ (see right part in Fig. \ref{fig:showline}), thus $z$ has the opposite sign of $a_2$. Since the choice of $x$, $y$, and $\varepsilon$ is arbitrary, the sign of $z$ only depends on the sign of $y$ and $a_2$, leading to our conclusion.
\end{proof}

\begin{figure}[!t]
	\centering
	\includegraphics[width=1\linewidth]{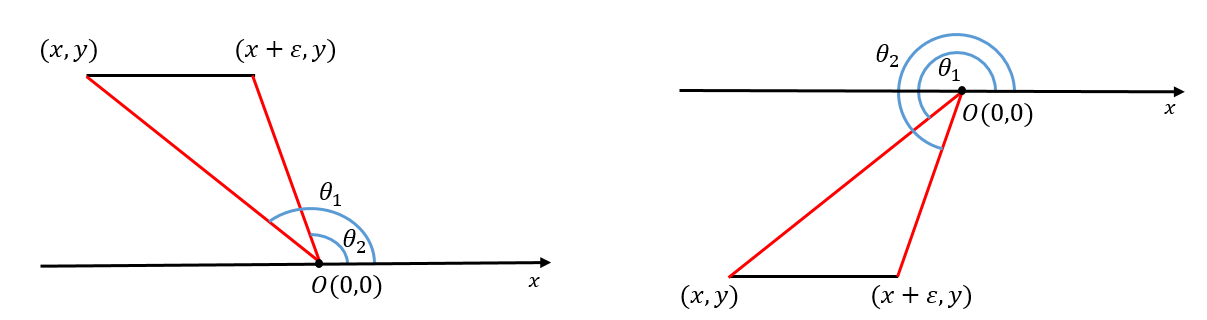}
	\caption{(left) Situation of $y>0$; (right) Situation of $y<0$.}
	\label{fig:showline}
\end{figure}

Now we prove the Property 1 and Property 2.
\subsubsection{Proof of Property 1}
\begin{proof}
We can establish this property by considering the weights of the edges in a grid digraph, which correspond to the rotation angles between the vectors at their respective endpoints. Given that all edges in a grid digraph are of the same length, by Lemma \ref{main}, as we approach the singularity, the field vector on an edge undergoes a faster rotation. Consequently, the rotation angle between the vectors at both ends of an edge increases as the edge gets closer to the singularity. Therefore, the weight of an edge is greater when it is closer to the singularity. Hence, the edge with the locally greatest weight must be part of the smallest square containing the singularity.
\end{proof}

\subsubsection{Proof of Property 2}
\begin{proof}
Before adding the edge with the greatest weight, according to Lemma \ref{main}, the representation containing the singularity must be either a \begin{tikzcd} {} \arrow[d] & {} \arrow[l] & {} \arrow[l] \\ {} \arrow[r] & {} \arrow[r] & {} \arrow[u] \end{tikzcd} 
or a \begin{tikzcd} {} \arrow[d] & {} \arrow[l] \\ {} \arrow[d] & {} \arrow[u] \\ {} \arrow[r] & {} \arrow[u] \end{tikzcd} 
configuration. Therefore, the edge with the greatest weight that is added last must be positioned in the middle of this hexagon. Regardless of the direction of the added edge, a new representation of the minimum generator of $H_1(G)$ emerges. Specifically, it is the smallest square that contains the singularity, taking the shape of \begin{tikzcd} {} \arrow[d] & {} \arrow[l] \\ {} \arrow[r] & {} \arrow[u] \end{tikzcd}.
\end{proof}

\section{Some tables of supplementary data}
\begin{table*}[!t]
	\centering
	\caption{ Positioning Results of Magnetic Dip Poles and Errors }
	\label{tab:IGRF}
	\resizebox{0.75\textwidth}{!}{%
		\begin{tabular}{|c|c|c|c|c|c|c|}
			\hline 
			\multicolumn{7}{|c|}{ North dip pole } \\
			\hline 
			Year & \begin{tabular}{c} 
				Positioning \\
				longitude
			\end{tabular} & \begin{tabular}{c} 
				Positioning \\
				latitude
			\end{tabular} & \begin{tabular}{c} 
				Actual \\
				longitude
			\end{tabular} & \begin{tabular}{c} 
				Actual \\
				latitude
			\end{tabular} & \begin{tabular}{c} 
				Errors of \\
				longitude
			\end{tabular} & \begin{tabular}{c} 
				Errors of \\
				latitude
			\end{tabular} \\
			\hline 2000 & 109.60$^{\circ}$W & 80.91$^{\circ}$N & 109.64$^{\circ}$W & 80.97$^{\circ}$N & 0.04$^{\circ}$ & 0.06$^{\circ}$ \\
			\hline 2005 & 118.13$^{\circ}$W & 83.13$^{\circ}$N & 118.22$^{\circ}$W & 83.19$^{\circ}$N & 0.09$^{\circ}$ & 0.06$^{\circ}$ \\
			\hline 2010 & 132.88$^{\circ}$W & 85.12$^{\circ}$N & 132.84$^{\circ}$W & 85.02$^{\circ}$N & 0.04$^{\circ}$ & 0.1$^{\circ}$ \\
			\hline 2015 & 160.37$^{\circ}$W & 86.37$^{\circ}$N & 160.34$^{\circ}$W & 86.31$^{\circ}$N & 0.03$^{\circ}$ & 0.06$^{\circ}$ \\
			\hline 2020 & 162.87$^{\circ}$E & 86.39$^{\circ}$N & 162.87$^{\circ}$E & 86.49$^{\circ}$N & 0$^{\circ}$ & 0.1$^{\circ}$ \\
			\hline 
			\multicolumn{7}{|c|}{ South dip pole } \\
			\hline 
			Year &\begin{tabular}{c} 
				Positioning \\
				longitude
			\end{tabular} & \begin{tabular}{c} 
				Positioning \\
				latitude
			\end{tabular} & \begin{tabular}{c} 
				Actual \\
				longitude
			\end{tabular} & \begin{tabular}{c} 
				Actual \\
				latitude
			\end{tabular} & \begin{tabular}{c} 
				Errors of \\
				longitude
			\end{tabular} & \begin{tabular}{c} 
				Errors of \\
				latitude
			\end{tabular} \\
			\hline 2000 & 138.36$^{\circ}$E & 64.64$^{\circ}$S & 138.30$^{\circ}$E & 64.66$^{\circ}$S & 0.06$^{\circ}$ & 0.02$^{\circ}$ \\
			\hline 2005 & 137.88$^{\circ}$E & 64.60$^{\circ}$S & 137.85$^{\circ}$E & 64.55$^{\circ}$S & 0.03$^{\circ}$ & 0.05$^{\circ}$ \\
			\hline 2010 & 137.36$^{\circ}$E & 64.39$^{\circ}$S & 137.32$^{\circ}$E & 64.43$^{\circ}$S & 0.04$^{\circ}$ & 0.04$^{\circ}$ \\
			\hline 2015 & 136.63$^{\circ}$E & 64.34$^{\circ}$S & 136.60$^{\circ}$E & 64.28$^{\circ}$S & 0.03$^{\circ}$ & 0.06$^{\circ}$ \\
			\hline 2020 & 135.88$^{\circ}$E & 64.12$^{\circ}$S & 135.87$^{\circ}$E & 64.08$^{\circ}$S & 0.01$^{\circ}$ & 0.04$^{\circ}$ \\
			\hline
		\end{tabular}
	}
\end{table*}

\begin{table*}[!t]
	\centering
	\caption{Positioning Results and Errors of the Center Tracking of Typhoon Saola}
        \label{tab:SL}
	\resizebox{0.8\textwidth}{!}{% 
		\begin{tabular}{|c|c|c|c|c|c|c|}
			\hline \text { UTC } & $\begin{array}{c}
				\text { Actual } \\
				\text { longitude(E) }
			\end{array}$ & $\begin{array}{c}
				\text { Actual } \\
				\text { latitude(N) }
			\end{array}$ & $\begin{array}{c}
				\text { Positioning } \\
				\text { longitude(E) }
			\end{array}$ & $\begin{array}{c}
				\text { Positioning } \\
				\text { latitude(N) }
			\end{array}$ & $\begin{array}{c}
				\text { Errors of } \\
				\text { longitude }
			\end{array}$ & $\begin{array}{c}
				\text { Errors of } \\
				\text { latitude }
			\end{array}$ \\
			\hline 2023/08/25/00:00 & 124   & 19.7  & 124.26 & 19.99 & 0.26  & 0.29 \\
			\hline 2023/08/25/06:00 & 123.9 & 19.7  & 124.25 & 19.52 & 0.35  & 0.18 \\
			\hline 2023/08/25/12:00 & 123.6 & 19.5  & 123.96 & 19.26 & 0.36  & 0.24 \\
			\hline 2023/08/25/18:00 & 123.4 & 19.3  & 123.53 & 19.22 & 0.13  & 0.08 \\
			\hline 2023/08/26/00:00 & 123.2 & 18.5  & 123.26 & 18.71  & 0.06  & 0.21 \\
			\hline 2023/08/26/06:00 & 123.2 & 18    & 123.28 & 18    & 0.08  & 0 \\
			\hline 2023/08/26/12:00 & 123.1 & 17.6  & 123.02 & 17.52 & 0.08  & 0.08 \\
			\hline 2023/08/26/18:00 & 122.9 & 17.2  & 123.04 & 17.24 & 0.14  & 0.04 \\
			\hline 2023/08/27/00:00 & 122.9 & 16.8  & 123   & 16.77 & 0.1   & 0.03 \\
			\hline 2023/08/27/06:00 & 123   & 16.5  & 122.97 & 16.49 & 0.03  & 0.01 \\
			\hline 2023/08/27/12:00 & 123.2 & 16.3  & 123.28 & 16.23 & 0.08  & 0.07 \\
			\hline 2023/08/27/18:00 & 123.8 & 16    & 123.74 & 16.21 & 0.06  & 0.21 \\
			\hline 2023/08/28/00:00 & 124.3 & 16.8  & 124.27 & 16.51 & 0.03  & 0.29 \\
			\hline 2023/08/28/06:00 & 124.3 & 17.5  & 124.26 & 17.48 & 0.04  & 0.02 \\
			\hline 2023/08/28/12:00 & 124.1 & 17.8  & 124.25 & 17.75 & 0.15  & 0.05 \\
			\hline 2023/08/28/18:00 & 123.9 & 18.2  & 124   & 18.22 & 0.1   & 0.02 \\
			\hline 2023/08/29/00:00 & 123.5 & 18.6  & 123.51 & 18.51 & 0.01  & 0.09 \\
			\hline 2023/08/29/06:00 & 123.1 & 18.9  & 123   & 18.76 & 0.1   & 0.14 \\
			\hline 2023/08/29/12:00 & 122.7 & 19.3  & 122.51 & 19.29 & 0.19  & 0.01 \\
			\hline 2023/08/29/18:00 & 121.9 & 19.9  & 121.96 & 19.73 & 0.06  & 0.17 \\
			\hline 2023/08/30/00:00 & 121.1 & 20.1  & 121.23 & 20    & 0.13  & 0.1 \\
			\hline 2023/08/30/06:00 & 120.4 & 20.4  & 120.46 & 20.49 & 0.06  & 0.09 \\
			\hline 2023/08/30/12:00 & 119.7 & 20.7  & 119.51 & 20.71 & 0.19  & 0.01 \\
			\hline 2023/08/30/18:00 & 118.7 & 20.9  & 118.99 & 20.96 & 0.29  & 0.06 \\
			\hline 2023/08/31/00:00 & 118.2 & 21    & 118.22 & 21    & 0.02  & 0 \\
			\hline 2023/08/31/06:00 & 117.8 & 21.2  & 117.74 & 21.26 & 0.06  & 0.06 \\
			\hline 2023/08/31/12:00 & 117.3 & 21.5  & 117.25 & 21.47 & 0.05  & 0.03 \\
			\hline 2023/08/31/18:00 & 116.8 & 21.7 & 116.99 & 21.54 & 0.19 &  0.16 \\
			\hline 2023/09/01/00:00 & 116.3 & 21.9  & 116.22 & 21.97 & 0.08  & 0.07 \\
			\hline 2023/09/01/06:00 & 115.5 & 22    & 115.5 & 21.78 & 0     & 0.22 \\
			\hline 2023/09/01/12:00 & 114.5 & 22    & 114.52 & 22.01 & 0.02  & 0.01 \\
			\hline 2023/09/01/18:00 & 113.5 & 22    & 113.49 & 21.77 & 0.01  & 0.23 \\
			\hline 	
		\end{tabular}
	}
\end{table*}

\begin{table*}[!t]
	\centering
	\caption{Positioning Results and Errors of the Center Tracking of Typhoon Khanun}
	\label{tab:KN}
	\resizebox{0.8\textwidth}{!}{% 
		\begin{tabular}{|c|c|c|c|c|c|c|}
			\hline \text { UTC } & $\begin{array}{c}
				\text { Actual } \\
				\text { longitude(E) }
			\end{array}$ & $\begin{array}{c}
				\text { Actual } \\
				\text { latitude(N) }
			\end{array}$ & $\begin{array}{c}
				\text { Positioning } \\
				\text { longitude(E) }
			\end{array}$ & $\begin{array}{c}
				\text { Positioning } \\
				\text { latitude(N) }
			\end{array}$ & $\begin{array}{c}
				\text { Errors of } \\
				\text { longitude }
			\end{array}$ & $\begin{array}{c}
				\text { Errors of } \\
				\text { latitude }
			\end{array}$ \\
			\hline 2023/07/31/00:00 & 132   & 22.1  & 131.99 & 22.22 & 0.01  & 0.12 \\
			\hline 2023/07/31/06:00 & 131.5 & 22.8  & 131.52 & 23    & 0.02  & 0.2 \\
			\hline 2023/07/31/12:00 & 131.1 & 23.4  & 130.99 & 23.29 & 0.11  & 0.11 \\
			\hline 2023/07/31/18:00 & 130.3 & 24    & 130.03 & 24.25 & 0.27  & 0.25 \\
			\hline 2023/08/01/00:00 & 129.4 & 24.6  & 129.27 & 24.51 & 0.13  & 0.09 \\
			\hline 2023/08/01/06:00 & 128.7 & 25    & 128.52 & 24.97 & 0.18  & 0.03 \\
			\hline 2023/08/01/12:00 & 128   & 25.3  & 127.77 & 25.22 & 0.23  & 0.08 \\
			\hline 2023/08/01/18:00 & 127.4 & 25.6  & 127.26 & 25.5  & 0.14  & 0.1 \\
			\hline 2023/08/02/00:00 & 126.8 & 25.7  & 126.78 & 25.75 & 0.02  & 0.05 \\
			\hline 2023/08/02/06:00 & 126.1 & 26    & 126   & 26.03 & 0.1   & 0.03 \\
			\hline 2023/08/02/12:00 & 125.6 & 26.2  & 125.52 & 26.02 & 0.08  & 0.18 \\
			\hline 2023/08/02/18:00 & 125.1 & 26.4  & 125.03 & 26.28 & 0.07  & 0.12 \\
			\hline 2023/08/03/00:00 & 124.8 & 26.7  & 124.77 & 26.53 & 0.03  & 0.17 \\
			\hline 2023/08/03/06:00 & 124.3 & 26.7  & 124.3 & 26.74 & 0     & 0.04 \\
			\hline 2023/08/03/12:00 & 124.1 & 26.9  & 124.26 & 26.72 & 0.16  & 0.18 \\
			\hline 2023/08/03/18:00 & 124.2 & 26.8  & 124.26 & 26.78 & 0.06  & 0.02 \\
			\hline 2023/08/04/00:00 & 124.5 & 27    & 124.72 & 27.02 & 0.22  & 0.02 \\
			\hline 2023/08/04/06:00 & 125.2 & 27.2  & 125.21 & 27.25 & 0.01  & 0.05 \\
			\hline 2023/08/04/12:00 & 125.7 & 27.5  & 125.73 & 27.49 & 0.03  & 0.01 \\
			\hline 2023/08/04/18:00 & 126.3 & 27.6  & 126.23 & 27.55 & 0.07  & 0.05 \\
			\hline 2023/08/05/00:00 & 126.9 & 27.7  & 126.75 & 27.77 & 0.15  & 0.07 \\
			\hline 2023/08/05/06:00 & 127.6 & 27.8  & 127.75 & 27.97 & 0.15  & 0.17 \\
			\hline 2023/08/05/12:00 & 128.4 & 27.9  & 128.47 & 27.78 & 0.07  & 0.12 \\
			\hline 2023/08/05/18:00 & 129   & 27.9  & 128.96 & 28.22 & 0.04  & 0.32 \\
			\hline 2023/08/06/00:00 & 129.4 & 27.8  & 129.5 & 28.2  & 0.1   & 0.4 \\
			\hline 2023/08/06/06:00 & 130   & 27.7  & 129.98 & 27.51 & 0.02  & 0.19 \\
			\hline 2023/08/06/12:00 & 130.5 & 27.5  & 130.3 & 27.46 & 0.2   & 0.04 \\
			\hline 2023/08/06/18:00 & 130.8 & 27.7  & 130.77 & 27.72 & 0.03  & 0.02 \\
			\hline 2023/08/07/00:00 & 131   & 27.8  & 131.25 & 27.75 & 0.25  & 0.05 \\
			\hline 2023/08/07/06:00 & 131.1 & 27.9  & 131.26 & 27.98 & 0.16  & 0.08 \\
			\hline 2023/08/07/12:00 & 131.3 & 28.2  & 131   & 28.24 & 0.3   & 0.04 \\
			\hline 2023/08/07/18:00 & 131.1 & 28.4  & 130.75 & 28.54 & 0.35  & 0.14 \\
			\hline 2023/08/08/00:00 & 130.9 & 29    & 131.04 & 29.24 & 0.14  & 0.24 \\
			\hline 2023/08/08/06:00 & 130.8 & 29.3  & 130.78 & 29.53 & 0.02  & 0.23 \\
			\hline 2023/08/08/12:00 & 130.4 & 29.9  & 130   & 29.78 & 0.4   & 0.12 \\
			\hline 2023/08/08/18:00 & 129.7 & 30.4  & 129.29 & 30.53 & 0.41  & 0.13 \\
			\hline 	
		\end{tabular}
	}
\end{table*}

\end{appendices}

\end{document}